    \documentclass[journal]{IEEEtran}
    
    \usepackage[utf8]{inputenc}
    \usepackage[english]{babel}
    \usepackage[T1]{fontenc}
    \usepackage{amsfonts}
    \usepackage{amsmath}
    \usepackage{amssymb}
    \usepackage{amsthm}
    \usepackage{tikz}
    \usepackage{graphicx}
    \usepackage{color}
    \usepackage{pgfpages}
    \usepackage{stmaryrd}
    \usepackage{xstring}
    \usepackage[ruled, vlined, linesnumbered]{algorithm2e}
    \usepackage{pgfplots, pgfplotstable}
    \usepackage{subcaption}
    \usepackage{hyperref}
    \usepackage{verbatim}
    
    \usetikzlibrary{calc, graphs}
    \pgfplotsset{compat=newest}

    \newcommand{\etal}{\emph{et al.}}
    \newcommand{\ie}{\emph{i.e.}}
    
    \newcommand{\eg}{\emph{e.g.}}
    \renewcommand{\th}{^\text{th}}
    
    \newcommand{\figref}[1]{Figure~\ref{#1}}
    \newcommand{\secref}[1]{Section~\ref{#1}}
    \newcommand{\algref}[1]{Algorithm~\ref{#1}}
    \newcommand{\propref}[1]{Proposition~\ref{#1}}
    
    \newcommand{\defref}[1]{Definition~\ref{#1}}
    \newcommand{\corref}[1]{corollary~\ref{#1}}
    \newcommand{\lemref}[1]{Lemma~\ref{#1}}
    \newcommand{\conjref}[1]{Conjecture~\ref{#1}}

    \makeatletter\newcommand\footnoteref[1]{\protected@xdef\@thefnmark{\ref{#1}}\@footnotemark}\makeatother
    
    \theoremstyle{definition}
    \newtheorem{definition}{Definition}
    \newtheorem{proposition}{Proposition}
    \newtheorem{corollary}{Corollary}
    \newtheorem{remark}{Remark}
    \newtheorem{lemma}{Lemma}
    \newtheorem{conjecture}{Conjecture}
    
    \SetKwProg{Function}{function}{}{}
    \newcommand{\State}[1]{#1 \\}
    \SetKw{KwAnd}{and}
    \SetKw{KwSt}{s.t.}
    \SetKw{KwTrue}{true}
    \SetKw{KwFalse}{false}
    \SetKw{KwNull}{NULL}
    \SetKw{KwBreak}{break}
    \SetKw{KwTo}{to}
    \SetKw{KwNot}{not}
    \SetKwRepeat{Do}{do}{while}

    \newcommand{\constantName}[1]{\text{$\MakeUppercase{#1}$}}
    \newcommand{\mathSetName}[1]{\text{$\mathbb{\MakeUppercase{#1}}$}}
    \newcommand{\setName}[1]{\text{$\mathcal{\MakeUppercase{#1}}$}}
    \newcommand{\variableName}[1]{\text{$\MakeLowercase{#1}$}}
    \newcommand{\vectorName}[1]{\text{$\mathbf{\MakeLowercase{#1}}$}}
    \newcommand{\vectorSymbolName}[1]{\text{$\mathbf{\boldsymbol{#1}}$}}
    \newcommand{\matrixName}[1]{\text{$\mathbf{\MakeUppercase{#1}}$}}
    
    \newcommand{\functionName}[1]{\text{$\MakeLowercase{#1}$}}
    
    \newcommand{\interval}[1]{\text{$\left[ #1 \right]$}}
    \newcommand{\set}[1]{\text{$\left\{ #1 \right\}$}}
    \newcommand{\condSet}[2]{\text{$\left\{ #1 \;\middle|\; #2 \right\}$}}
    \newcommand{\tuple}[1]{\text{$\langle #1 \rangle$}}
    \newcommand{\sequence}[1]{\text{$\left( #1 \right)$}}
    \newcommand{\intInterval}[1]{\text{$\left\llbracket #1 \right\rrbracket$}}
    \newcommand{\cardinal}[1]{\text{$\left| #1 \right|$}}
    \newcommand{\absoluteValue}[1]{\text{$\left| #1 \right|$}}
    \newcommand{\norm}[2]{\text{$\left\| #1 \right\|_{#2}$}}
    \newcommand{\entry}[1]{\text{$\left[ #1 \right]$}}
    \renewcommand{\vector}[1]{\text{$\protect\begin{bmatrix} #1 \end{bmatrix}$}}
    \newcommand{\neighborhood}[2]{\text{$\setName{N}_{#1}(#2)$}}
    \newcommand{\positiveIntegers}{\text{$\mathSetName{N}$}}
    \newcommand{\strictlyPositiveIntegers}{\text{$\mathSetName{N}^*$}}
    \newcommand{\integers}{\text{$\mathSetName{Z}$}}
    \newcommand{\reals}{\text{$\mathSetName{R}$}}
    \newcommand{\integersModulo}[1]{\text{$\integers/#1\integers$}}
    
    \newcommand{\lNorm}[1]{\text{$\ell_{#1}$}}
    \newcommand{\approximate}[1]{\text{$\widetilde{#1}$}}
    \newcommand{\complexity}[1]{\text{$\mathcal{O}\left(#1\right)$}}
    
    \newcommand{\graph}{\text{$\setName{G}$}}
    \newcommand{\digraph}[1]{\text{${\overrightarrow{\graph}}^{#1}$}}
    \newcommand{\vertices}{\text{$\setName{V}$}}
    \newcommand{\edges}{\text{$\setName{E}$}}
    \newcommand{\diedges}[1]{\text{${\overrightarrow{\edges}}^{#1}$}}
    \newcommand{\vertex}{\text{$\variableName{v}$}}
    \newcommand{\vertexVector}{\text{$\vectorName{v}$}}
    \newcommand{\diracVector}{\text{$\vectorName{e}$}}
    \newcommand{\edge}[2]{\text{$\set{#1, #2}$}}
    \newcommand{\diedge}[2]{\text{$\sequence{#1, #2}$}}
    \newcommand{\graphOrder}{\text{$\constantName{N}$}}
    \newcommand{\constant}{\text{$\constantName{C}$}}
    \newcommand{\radius}{\text{$\constantName{R}$}}

    \newcommand{\adjacency}{\text{$\matrixName{A}$}}

    \newcommand{\signalVector}{\text{$\vectorName{x}$}}
    
    \newcommand{\gridDimensions}{\text{$\vectorName{d}$}}
    \newcommand{\gridNbDimensions}{\text{$\constantName{D}$}}
    \renewcommand{\i}{\text{$\variableName{i}$}}
    \renewcommand{\j}{\text{$\variableName{j}$}}
    \renewcommand{\k}{\text{$\variableName{k}$}}
    \newcommand{\blackHole}{\text{$\bot$}}
    
    \newcommand{\transformation}{\text{$\functionName{\phi}$}}
    \newcommand{\losslessTransformation}{\text{$\functionName{\phi}^*$}}
    \newcommand{\transformations}[1]{\text{$\Phi_{#1}$}}
    \newcommand{\losslessTransformations}[1]{\text{$\Phi^*_{#1}$}}
    \newcommand{\translation}{\text{$\functionName{\psi}$}}
    \newcommand{\losslessTranslation}{\text{$\functionName{\psi}^*$}}
    \newcommand{\differenceVector}{\text{$\vectorSymbolName{\delta}$}}
    
    \newcommand{\translations}[1]{\text{$\Psi_{#1}$}}
    \newcommand{\losslessTranslations}[1]{\text{$\Psi^*_{#1}$}}
    \newcommand{\geodesic}{\text{$\functionName{d}$}}
    \newcommand{\score}{\text{$\functionName{s}$}}
    \newcommand{\loss}{\text{$\functionName{loss}$}}
    \newcommand{\EC}[1]{\text{EC$\left(\transformations{#1}\right)$}}
    \newcommand{\WNP}[1]{\text{WNP$\left(\transformations{#1}\right)$}}
    \newcommand{\SNP}[1]{\text{SNP$\left(\transformations{#1}\right)$}}
    \newcommand{\ISO}[1]{\text{ISO$\left(\transformations{#1}\right)$}}
    \newcommand{\losslessEC}[1]{\text{EC$\left(\losslessTransformations{#1}\right)$}}
    \newcommand{\losslessWNP}[1]{\text{WNP$\left(\losslessTransformations{#1}\right)$}}
    \newcommand{\losslessSNP}[1]{\text{SNP$\left(\losslessTransformations{#1}\right)$}}
    \newcommand{\losslessISO}[1]{\text{ISO$\left(\losslessTransformations{#1}\right)$}}

    \title{A neighborhood-preserving\\translation operator on graphs}

    \author
    {
        Bastien Pasdeloup$^\dagger$,
        Vincent Gripon$^\ddagger$,
        Jean-Charles Vialatte$^\ddagger$, \\
        Nicolas Grelier$^\ddagger$,
        Dominique Pastor$^\ddagger$
        \thanks{$\dagger$ Ecole Polytechnique de Lausanne, EPFL STI IEL LTS4, Station 11, CH-1015 Lausanne, Switzerland. Email: \{name.surname\}@epfl.ch.}
        \thanks{$\ddagger$ UMR CNRS Lab-STICC, IMT Atlantique, Technopôle Brest-Iroise, 29238 Brest Cedex 03, France. Email: \{name.surname\}@imt-atlantique.fr.}
    }
    
\begin{document}
        
            
            \maketitle
            
            \begin{abstract}
                In this paper, we introduce translation operators on graphs.
                Contrary to spectrally-defined translations in the framework of graph signal processing, our operators mimic neighborhood-preserving properties of translation operators defined in Euclidean spaces directly in the vertex domain, and therefore do not deform a signal as it is translated.
                We show that in the case of grid graphs built on top of a metric space, these operators exactly match underlying Euclidean translations, suggesting that they completely leverage the underlying metric.
                More generally, these translations are defined on any graph, and can therefore be used to process signals on those graphs.
                We show that identifying proposed translations is in general an NP-Complete problem.
                To cope with this issue, we introduce relaxed versions of these operators, and illustrate translation of signals on random graphs.
            \end{abstract}
            
            
            \section{Introduction}
            \label{introduction}
            
                Graph signal processing is a generalization of classical signal processing that arose a few years ago.
                The field developed around the observations that eigenvectors of a particular matrix --- the Laplacian matrix --- associated with a ring graph as depicted in \figref{ringGraph} can be chosen to correspond to classical discrete Fourier modes.
                In more details, a \emph{graph Fourier basis} associated with a graph of $\graphOrder$ vertices is a basis defined by its eigenvectors, and a spectral representation of any signal on a graph --- a vector in $\reals^\graphOrder$ --- can be obtained by projecting it into this particular basis, thus providing a \emph{spectral representation} for the signal.
                
                The correspondence between eigenvectors of the Laplacian matrix and the Fourier basis has been extended to any graph on which signals can be observed.
                Researchers have then successfully been able to find tools such as convolution, filtering, or modulation of signals on graphs (see \cite{shuman2013emerging, ortega2018graph} for an overview of such tools).
                
                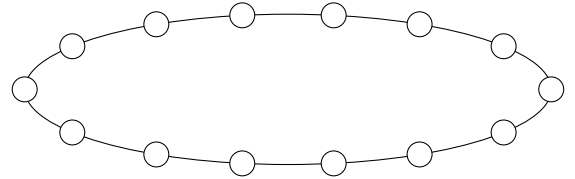
\begin{figure}
                    \centering
                    \begin{tikzpicture}
                        \draw (0,0) ellipse (3.5 and 1);
                        \node [draw, fill=white, circle] at ($(0,0)+(0:3.5 and 1)$) {};
                        \node [draw, fill=white, circle] at ($(0,0)+(35:3.5 and 1)$) {};
                        \node [draw, fill=white, circle] at ($(0,0)+(60:3.5 and 1)$) {};
                        \node [draw, fill=white, circle] at ($(0,0)+(80:3.5 and 1)$) {};
                        \node [draw, fill=white, circle] at ($(0,0)+(100:3.5 and 1)$) {};
                        \node [draw, fill=white, circle] at ($(0,0)+(120:3.5 and 1)$) {};
                        \node [draw, fill=white, circle] at ($(0,0)+(145:3.5 and 1)$) {};
                        \node [draw, fill=white, circle] at ($(0,0)+(180:3.5 and 1)$) {};
                        \node [draw, fill=white, circle] at ($(0,0)+(215:3.5 and 1)$) {};
                        \node [draw, fill=white, circle] at ($(0,0)+(240:3.5 and 1)$) {};
                        \node [draw, fill=white, circle] at ($(0,0)+(260:3.5 and 1)$) {};
                        \node [draw, fill=white, circle] at ($(0,0)+(280:3.5 and 1)$) {};
                        \node [draw, fill=white, circle] at ($(0,0)+(300:3.5 and 1)$) {};
                        \node [draw, fill=white, circle] at ($(0,0)+(325:3.5 and 1)$) {};
                    \end{tikzpicture}
                    \caption
                    {
                        Example of a ring graph.
                        Eigenvectors of the Laplacian matrix associated with the smallest eigenvalues correspond to the Fourier modes associated with the lowest frequencies in classical Fourier analysis.
                        The correspondence also holds as the eigenvalues increase.
                    }
                    \label{ringGraph}
                \end{figure}
                
                Among these tools, one of paramount importance is the translation operator, that allows one to move a signal on the graph.
                While understanding translation of temporal signals or images is straightforward due to the underlying vector space, it is not the case for signals on graphs in general, since such objects only consist of vertices and edges linking them, without any explicit mention of an underlying vector space.
                Multiple definitions of translations for signals on graphs have been proposed in the literature, but none has the property --- that we observe in classical signal processing --- that adjacent signal entries necessarily remain adjacent after translation.
                This has the effect to deform the signal as it is translated, either by breaking neighborhoods or by changing the signal energy, as illustrated in \figref{imgTranslation}.
                In this work, we introduce a novel translation definition that operates in the vertex domain, and enforces shape preservation of the translated signals.
                
                \begin{figure}
                    \centering
                    \begin{subfigure}{0.47\linewidth}
                        \includegraphics[width=\linewidth]{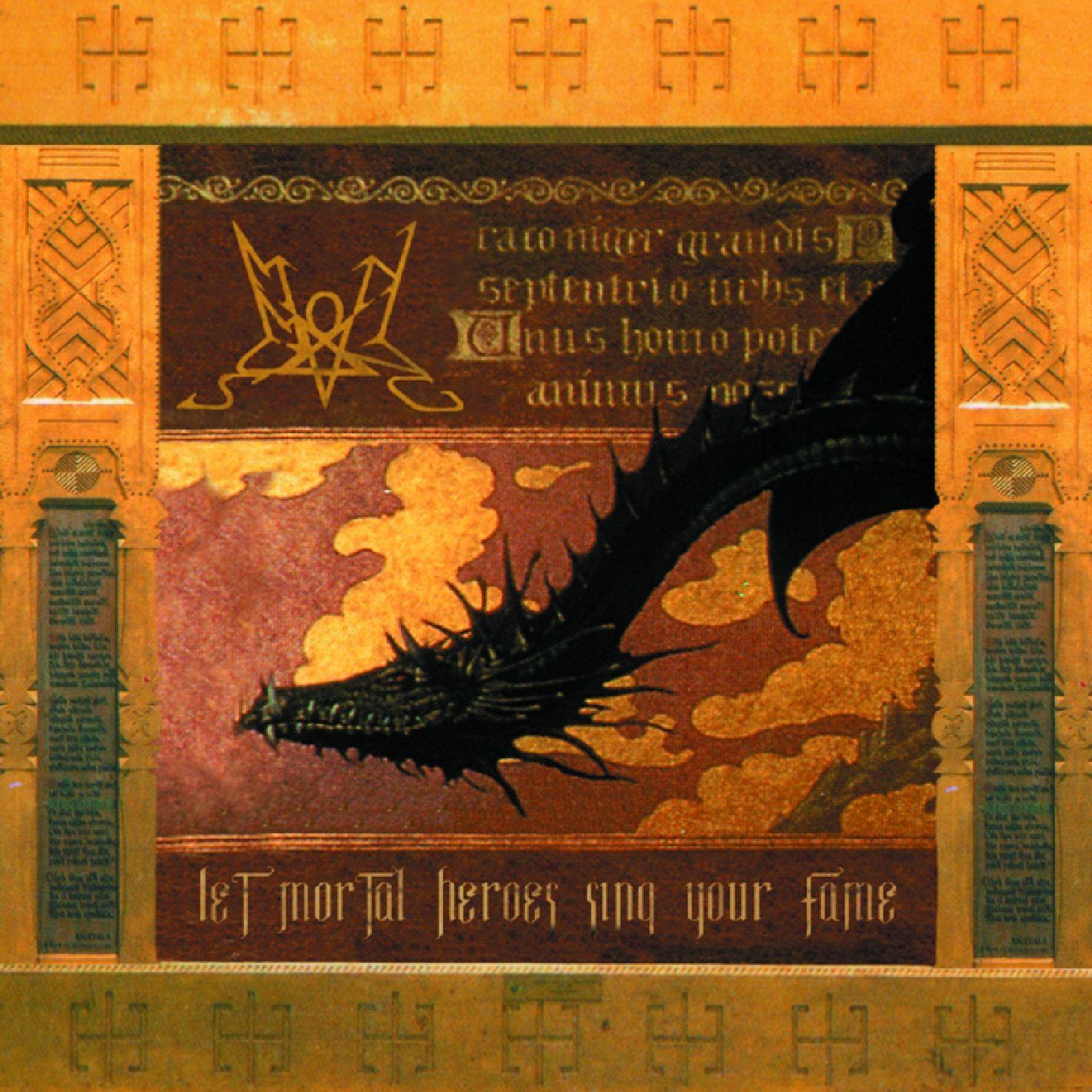}
                        \caption{}
                    \end{subfigure}
                    \hfill
                    \begin{subfigure}{0.47\linewidth}
                        \includegraphics[width=\linewidth]{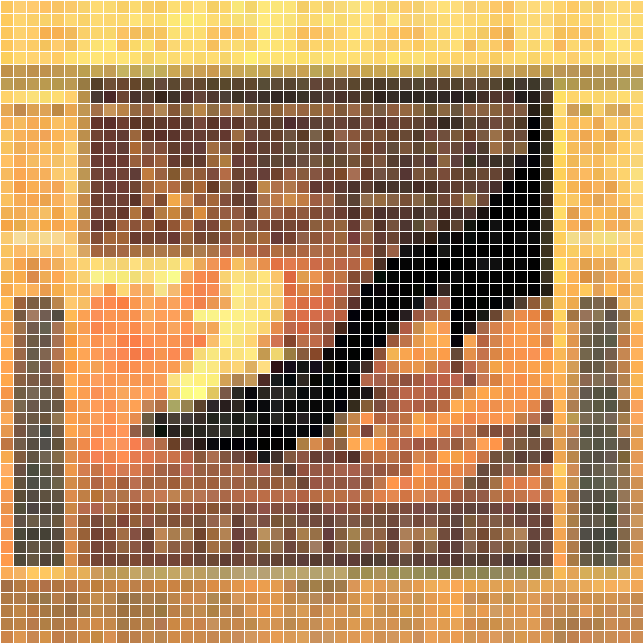}
                        \caption{}
                    \end{subfigure} \\
                    \vspace{0.2cm}
                    
                    \begin{subfigure}{0.47\linewidth}
                        \includegraphics[width=\linewidth]{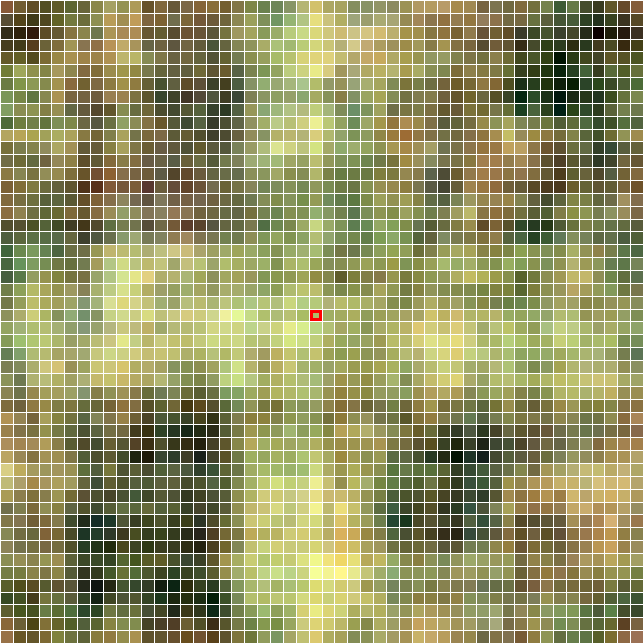}
                        \caption{}
                    \end{subfigure}
                    \hfill
                    \begin{subfigure}{0.47\linewidth}
                        \includegraphics[width=\linewidth]{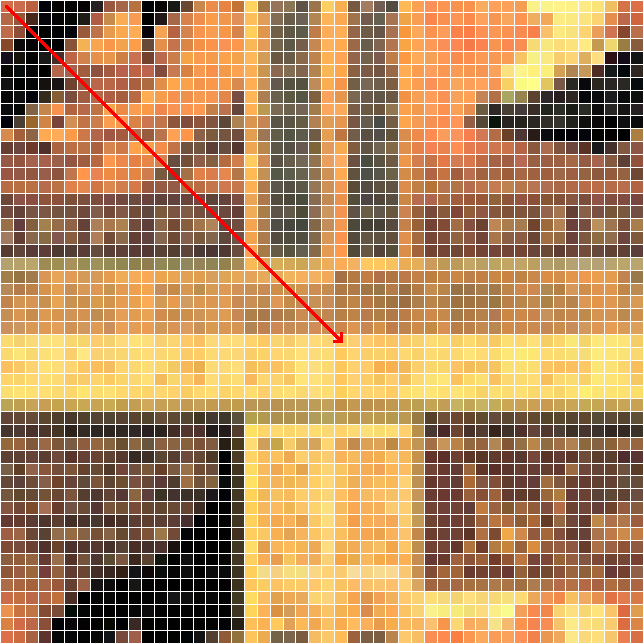}
                        \caption{}
                    \end{subfigure}
                    \caption
                    {
                        An image \textbf{(a)} that we want to translate, given as a signal of pixel intensity on a two-dimensional torus graph (see \figref{gridAndTorus}) \textbf{(b)}.
                        In the framework of graph signal processing, translation is achieved by convolving the initial image with a Dirac signal at a given location, highlighted in red in \textbf{(c)}.
                        The translation of the signal in \textbf{(b)} using this method is the image depicted in \textbf{(c)}.
                        When considering Euclidean translations, one would have expected to obtain the image in \textbf{(d)} after translation, or only its lower-right part in case of translation on a non-toric grid graph.
                    }
                    \label{imgTranslation}
                \end{figure}
                
                We advocate for the use of neighborhood-preserving translations because of three main reasons. First, it ensures obtained translations are isometries (cf.~\secref{definitions}), in the sense that they preserve distances. In many applications (\eg, machine learning), it thus ensures that patterns have similar shapes when they are shifted on the underlying graph.
                Second, it allows to draw a perfect match between Euclidean translations and graph translations when considering graphs that are defined using an underlying Euclidean metric space (cf.~\secref{resultsProofs}). As such, neighborhood-preserving translations can be seen as an extension of Euclidean translations.
                Finally, it preserves locality of shifted patterns on graphs, and as such can be used to extend locally coherent concepts that echo directly to classical domains, such as stationarity or uncertainty.
                
                In this article, we propose a framework to define translations of signals on arbitrary graphs that preserve adjacency, as translations of time signals or images would.
                Translations as we define them can be seen as an orientation of a subset of edges in the graph, with some neighborhood preservation constraints.
                Contrary to existing translation operators on graphs, we do not modify the signal entries as we translate them, with the exception of some special cases when we accept to lose part of the signal due to \emph{border effects}.
                Importantly, our translations do not rely on the use of any underlying metric space, which makes them applicable to any graph.
                
                The present article is organized as follows.
                First, \secref{related} recalls existing definitions for translations defined on graphs.
                In \secref{definitions}, we introduce neighborhood-preserving translation operators on graphs.
                \secref{resultsProofs} then gives results on these operators, such as the NP-completeness of the problem that consists in identifying them.
                To cope with this complexity issue, \secref{relaxation} then introduces a possible relaxation of the introduced operators.
                Finally, \secref{examples} depicts obtained translations on some graphs.

            \section{Related work}
            \label{related}
                
                \subsection{The graph shift approach}
                
                    Studying the case of the ring graph, Püschel and Moura \cite{pueschel2006}, followed by Sandrihaila and Moura \cite{sandryhaila2013discrete}, propose a notion of \emph{graph shift} as the adjacency matrix of the graph on which signals are defined.
                    In particular, when considering the directed ring graph --- \ie, the orientation of all edges of the ring graph in the same direction --- as a graph shift, multiplication of a signal by this shift has the effect to \emph{advance} it in time.
                    In the general case, considering an adjacency matrix as a translation operator has the effect to \emph{diffuse} a signal as it is translated.
                    Note that this is also the case where the adjacency matrix is normalized by its eigenvalue with the highest magnitude \cite{sandryhaila2014discrete}, in which case the signal energy only decreases as it is translated.
                    
                    For this reason, translation of signals with this approach cannot match our objective of conserving the signal entries during translation.
                    However, our approach is similar to this one in the sense that we identify translations by finding a subset of non-null entries of the adjacency matrix, which can be interpreted as an orientation of a subset of edges in the graph.
                    In particular, the directed ring graph is a valid translation on the ring graph according to our definitions.

                \subsection{The convolutive approach}
                
                    In the context of applying wavelets to graph signals, Hammond \etal{} \cite{hammond2011wavelets} propose to define translation as a localization function to move a wavelet to a particular location of the graph.
                    This is done by applying the wavelet to an \emph{impulse}, \ie, a signal that has all its energy concentrated at a single vertex.
                    
                    The same approach is taken by Shuman \etal{} \cite{shuman2012windowed, shuman2013emerging}, who propose a definition of translation of a signal \emph{to a vertex $\vertex$}, by convolution of this signal with an impulse located on $\vertex$.
                    This is done with analogy to the classical result in Fourier analysis that states that convolution in the time domain (in our case the graph) is equivalent to multiplication in the frequency domain (in our case the spectral domain of the graph).
                    
                    With this approach, the signal is moved to a particular location rather than by a certain quantity.
                    The convolution operation does not take the neighborhood in consideration, and allows modification of the signal when translating it.

                \subsection{The isometric approach}
                
                    Girault \etal{} \cite{girault2015translation, Girault2015a} propose a translation operator for graphs that is isometric with respect to the $\lNorm{2}$ norm, \ie, that does not change the signal energy as it is translated.
                    Their approach consists in changing the phase of the signal in the spectral domain to move it in the graph domain.
                    Additionally to keeping the signal norm unchanged, this operator has the property to preserve the signal localization, \ie, to have its energy located around a target vertex \cite{girault2016localization}.
                    
                    This approach can also be considered as convolutive, since the translation is performed by convolving the signal with complex exponentials.
                    Therefore, it suffers from the same drawback as the method introduced before, and can transform the signal while translating it.
                    
                    Gavili and Zhang \cite{2015arXiv151103512G} take a similar direction, and also propose a phase change for translation.
                    Contrary to the approach of Girault \etal{}, their solution does not take the graph spectrum into consideration.
                    Again, this method does not have any neighboring preservation property.

                \subsection{Neighborhood-preserving translations}
                
                    This article is an extended version of \cite{grelier2016neighborhood}.
                    In this work, authors explored translations on grid and torus graphs, and showed that Euclidean translations of images are equivalent to neighborhood preserving properties on these graphs.
                    It is worth noting that translations introduced in that first work were used in \cite{pasdeloup2017extending, lassance2018matching} in order to propose a generalization of convolutional neural networks to irregular domains.
                    
                    In the present article, we first reformulate the results in \cite{grelier2016neighborhood} to make them more general and comprehensive.
                    Additionally, we provide properties of the translations we propose, and show that identifying them is an NP-complete problem.
                    Then, we propose a relaxation of translations to address the complexity issue, and illustrate on random graphs.

            \section{Definitions}
            \label{definitions}
                
                In this section we introduce important definitions.
                After presenting some families of graphs --- that will be useful for demonstrations and examples --- we introduce a definition of transformations and translations of a signal on a graph.
                Connections to intuitive translations on an Euclidean space are made in \secref{resultsProofs}.
                
                \subsection{Some families of graphs}

                    \begin{definition}[Graph]
                        A \emph{graph} is a tuple $\graph = \tuple{\vertices, \edges}$, where $\vertices$ is the set of vertices and $\edges \subset \binom{\vertices}{2}$ is the set of edges\footnote{$\binom{\vertices}{2}$ denotes the set of unordered pairs of distinct elements in $\vertices$.}.
                        Graphs defined this way are by construction \emph{simple} (\ie, $\forall \vertex_1, \vertex_2 \in \vertices : \edge{\vertex_1}{\vertex_2} \in \edges \Rightarrow \vertex_1 \neq \vertex_2$) and \emph{undirected} (\ie, $\forall \vertex_1, \vertex_2 \in \vertices : \edge{\vertex_1}{\vertex_2} = \edge{\vertex_2}{\vertex_1}$).
                        \label{graph}
                    \end{definition}
                    
                    \begin{definition}[Digraph]
                        A \emph{digraph}, or \emph{directed graph}, is a tuple $\digraph{} = \tuple{\vertices, \diedges{}}$, where $\vertices$ is the set of vertices and $\diedges{} \subset \vertices \times \vertices$ is the set of directed edges, or \emph{diedges}.
                        Contrary to (undirected) graphs, the order of the vertices in elements of $\diedges{}$ matters.
                        In this article, we consider \emph{simple} digraphs only.
                    \end{definition}
                    
                    A digraph can be seen as a graph, from which some edges have been oriented to a particular direction:
                    
                    \begin{definition}[Orientation of a graph]
                        Let $\graph = \tuple{\vertices, \edges}$.
                        An orientation of $\graph$ is a digraph $\digraph{} = \tuple{\vertices, \diedges{}}$ such that $\forall \diedge{\vertex_1}{\vertex_2} \in \diedges{} : \edge{\vertex_1}{\vertex_2} \in \edges$.
                    \end{definition}
                    
                    It is often preferred to index vertices from $1$ to $\graphOrder = \cardinal{\vertices}$, with $\cardinal{\cdot}$ being the cardinality operator.
                    Thus, without loss of generality, and denoting $\intInterval{\constant_1, \constant_2}$ the set of all integers between $\constant_1$ and $\constant_2$, both included, we suppose in the remaining of this work that $\vertices = \intInterval{1, \graphOrder}$.
                    This allows us to define the adjacency matrix of a graph:
                    
                    \begin{definition}[Adjacency matrix]
                        A \emph{(binary) adjacency matrix} $\adjacency$ for a graph $\graph = \tuple{\vertices, \edges}$ is a $\graphOrder \times \graphOrder$ square matrix with:
                        $$
                            \forall \vertex_1, \variableName{v}_2 \in \vertices :
                            \adjacency\entry{\vertex_1, \vertex_2} =
                            \left\{
                                \begin{array}{cl}
                                    1 & \text{if~} \edge{\vertex_1}{\vertex_2} \in \edges \\
                                    0 & \text{otherwise}
                                \end{array}
                            \right.
                            \;.
                        $$
                        \label{adjacencyMatrix}
                    \end{definition}
                    
                    In this article, we note $\adjacency\entry{\i, \j}$ the entry of matrix $\adjacency$ at row $\i$ and column $\j$.
                    Additionally, we use the notations $\adjacency\entry{\i, :}$ and $\adjacency\entry{:, \j}$ for the $\i\th$ row and $\j\th$ column of $\adjacency$, respectively.
                    
                    Note that \defref{adjacencyMatrix} holds for digraphs, but in that case the adjacency matrix is not necessarily symmetric.
                    This adjacency matrix provides a convenient way to represent adjacency between vertices.
                    
                    Finally, we note $\neighborhood{H}{\vertex_1} \subseteq \vertices$ the $H$-hop \emph{neighborhood} of a vertex $\vertex_1 \in \vertices$, \ie, the set of vertices $\vertex_2 \in \vertices$ such that $\geodesic(\vertex_1, \vertex_2) = H$, where $\geodesic$ is the geodesic distance on the graph \cite{moore1959shortest}.
                    
                    For the sake of demonstrations and examples in this article, we introduce some particular graphs:
                    
                    \begin{definition}[Complete graph]
                        The \emph{complete graph} $\graph_c = \tuple{\vertices_c, \edges_c}$ of order $\graphOrder$ is the graph such that:
                        $$
                            \edges_c = \binom{\vertices_c}{2}
                            \;.
                        $$
                    \end{definition}
                    
                    Grid graphs are a natural support for (1D) time or (2D) image signals.
                    We introduce them in the general case of dimension $\gridNbDimensions$:

                    \begin{definition}[Grid graph]
                        Let $\gridDimensions \in \strictlyPositiveIntegers^\gridNbDimensions$.
                        The \emph{grid graph} $\graph_g = \tuple{\vertices_g, \edges_g}$ yielded by the dimensions vector $\gridDimensions$ is the graph such that:
                        \begin{itemize}
                            \item $\vertices_g = \intInterval{1, \gridDimensions\entry{1}} \times \intInterval{1, \gridDimensions\entry{2}} \times \dots \times \intInterval{1, \gridDimensions\entry{\gridNbDimensions}}$;
                            \item $\forall \vertexVector_1, \vertexVector_2 \in \vertices_g : (\edge{\vertexVector_1}{\vertexVector_2} \in \edges_g) \Leftrightarrow (\exists \i \in \intInterval{1, \gridNbDimensions} : (\absoluteValue{\vertexVector_1\entry{\i} - \vertexVector_2\entry{\i}} = 1) \wedge (\forall \j \in \intInterval{1, \gridNbDimensions}, \j \neq \i : \vertexVector_1\entry{\j} = \vertexVector_2\entry{\j}))$.
                        \end{itemize}
                        \label{defGrid}
                    \end{definition}
                    
                    The torus graph of dimensions $\gridDimensions$ can be defined just as the grid graph, considering operations to be performed over $\integersModulo{\gridDimensions\entry{\i}}$ for the $\i\th$ coordinate:
                    
                    \begin{definition}[Torus graph]
                        Let $\gridDimensions \in \strictlyPositiveIntegers^\gridNbDimensions$.
                        The \emph{torus graph} $\graph_t = \tuple{\vertices_t, \edges_t}$ yielded by the dimensions vector $\gridDimensions$ is the graph such that:
                        \begin{itemize}
                            \item $\vertices_t = \intInterval{1, \gridDimensions\entry{1}} \times \intInterval{1, \gridDimensions\entry{2}} \times \dots \times \intInterval{1, \gridDimensions\entry{\gridNbDimensions}}$;
                            \item $\forall \vertexVector_1, \vertexVector_2 \in \vertices_t : (\edge{\vertexVector_1}{\vertexVector_2} \in \edges_t) \Leftrightarrow (\exists \i \in \intInterval{1, \gridNbDimensions} : (\absoluteValue{\vertexVector_1\entry{\i} - \vertexVector_2\entry{\i}} \in \set{1, \gridDimensions\entry{\gridNbDimensions} - 1}) \wedge (\forall \j \in \intInterval{1, \gridNbDimensions}, \j \neq \i : \vertexVector_1\entry{\j} = \vertexVector_2\entry{\j}))$.
                        \end{itemize}
                        \label{defTorus}
                    \end{definition}
                    
                    As an example, \figref{gridAndTorus} provides a visual representation of the grid graph and the torus graph that are yielded by the dimensions vector $\gridDimensions = \vector{6\\5}$.
                    Vertices are placed according to the coordinates associated with vertices in $\vertices_g$ and $\vertices_t$.
                    
                    \begin{figure}
                        \centering
                        \begin{tikzpicture}[scale=0.6, thick]
                          \tikzstyle{every node} = [draw, circle];
                          \foreach \i in {0,1,2,3,4,5}{
                            \foreach \j in {0,1,2,3,4}{
                              \node(\i\j) at (\i,\j) {};
                            }
                          }
                          \path[]
                          \foreach \i in {0,1,2,3,4,5}{
                            \foreach \j/\jj in {0/1,1/2,2/3,3/4}{
                              (\i\j) edge (\i\jj)
                            }
                          }
                          \foreach \i/\ii in {0/1,1/2,2/3,3/4,4/5}{
                            \foreach \j in {0,1,2,3,4}{
                              (\i\j) edge (\ii\j)
                            }
                          }
                          ;
                        \end{tikzpicture}
                        ~~~~~
                        \begin{tikzpicture}[scale=0.6,thick]
                          \tikzstyle{every node} = [draw, circle];
                          \foreach \i in {0,1,2,3,4,5}{
                            \foreach \j in {0,1,2,3,4}{
                              \node(\i\j) at (\i,\j) {};
                            }
                          }
                          \path[]
                          \foreach \i in {0,1,2,3,4,5}{
                            \foreach \j/\jj in {0/1,1/2,2/3,3/4}{
                              (\i\j) edge (\i\jj)
                            }
                          }
                          \foreach \i/\ii in {0/1,1/2,2/3,3/4,4/5}{
                            \foreach \j in {0,1,2,3,4}{
                              (\i\j) edge (\ii\j)
                            }
                          }
                          \foreach \i in {0,1,2,3,4,5}{
                            (\i4) edge[bend left] (\i0)
                          }
                          \foreach \j in {0,1,2,3,4}{
                            (0\j) edge[bend left] (5\j)
                          }
                          ;
                        \end{tikzpicture}
                        \caption[]{Example of the grid graph (left) and torus graph (right), both with dimensions $\gridDimensions = \vector{6\\5}$.}
                        \label{gridAndTorus}
                    \end{figure}
                    
                    It is important to notice that the grid graph and the torus graph are defined by associating coordinates with their vertices corresponding to a regular sampling of a portion of the Euclidean space.
                    Therefore, vertices of $\vertices_g$ and $\vertices_t$ are by construction isomorphic to vectors, and defining a notion of translation is natural.
                    However, if the connection between vertices and the underlying metric space is unknown, it is not obvious to retrieve the underlying translations (which, in this case, boils down to retrieving the connection between vertices and the underlying vectors).
                    We show in \secref{resultsProofs} that our proposed definitions of translations is one way of achieving it.
                    
                    When considering graphs in the general case, no underlying metric space is available, and only the existence of connections among vertices is observed.
                    Thus, there is no notion of translation to a particular direction due to the absence of an underlying Euclidean space, and one can only rely on the neighborhood of the vertices.
                    In other words, graphs are useful to encapsulate a notion of neighborhood, but they do not naturally convey a notion of directionality.
                    
                    Finally, one category of graphs we will consider in this article are random geometric graphs, defined as follows:
                    
                    \begin{definition}[Random geometric graph]
                        A graph $\graph_{rg} = \tuple{\vertices_{rg}, \edges_{rg}}$ of order $\graphOrder$ following a random geometric model of parameter $\radius$ is a graph such that every vertex is associated with a uniformly randomly drawn location in a two-dimensional unit square. Then, edges correspond to pairs of vertices whose locations are under a certain radius, \ie:
                        \begin{itemize}
                            \item $\vertices_{rg} = \set{\vertexVector_1, \dots, \vertexVector_\graphOrder}$, with $\forall \i \in \intInterval{1, \graphOrder} : \vertexVector_\i \in \interval{0, 1}^2$;
                            \item $\forall \vertexVector_1, \vertexVector_2 \in \vertices_{rg} : \edge{\vertexVector_1}{\vertexVector_2} \in \edges_{rg} \Leftrightarrow \norm{\vertexVector_1 - \vertexVector_2}{2} < \radius$, where $\norm{\cdot}{2}$ measures the $\lNorm{2}$ norm of a vector.
                        \end{itemize}
                    \end{definition}

                \subsection{Transformations and translations on graphs}
                
                    Let us consider a graph $\graph = \tuple{\vertices, \edges}$.
                    Additionally, let us introduce an element $\blackHole$ such that $\blackHole \notin \vertices$.
                    
                    \begin{definition}[Transformation]
                        A \emph{transformation} on a graph $\graph = \tuple{\vertices, \edges}$ is a function $\transformation : \vertices \cup \set{\blackHole} \to \vertices \cup \set{\blackHole}$ such that
                        \begin{itemize}
                            \item $\forall \vertex_1, \vertex_2 \in \vertices_{\not\blackHole} : (\transformation(\vertex_1) = \transformation(\vertex_2)) \Rightarrow (\vertex_1 = \vertex_2)$;
                            \item $\transformation(\blackHole) = \blackHole$;
                        \end{itemize}
                        where $\vertices_{\not\blackHole} = \{\vertex\in \vertices: \transformation(\vertex) \neq \blackHole\}$.
                        We denote the set of transformations on $\graph$ by $\transformations{\graph}$.
                    \end{definition}

                    Informally, a transformation $\transformation \in \transformations{\graph}$ on a graph is a function that is injective for every vertex whose image is not $\blackHole$.
                    Including $\blackHole$ in the definition domain allows a natural composition of transformations.

                    \begin{definition}[Loss of a transformation]
                        We call \emph{loss} of a transformation the quantity $\cardinal{\condSet{\vertex \in \vertices}{\transformation(\vertex) = \blackHole}}$, noted $\loss(\transformation)$.
                        In the case where $\loss(\transformation) = 0$, we say that $\transformation$ is \emph{lossless}, and we note it $\losslessTransformation$.
                        We denote the set of lossless transformations on $\graph$ by $\losslessTransformations{\graph}$.
                    \end{definition}
                    
                    In the case of lossless transformations, every vertex has an image in $\vertices$.
                    Therefore, they are bijective from $\vertices$ to $\vertices$.
                    
                    It is also interesting to notice that every graph $\graph = \tuple{\vertices, \edges}$ admits a transformation of loss $\graphOrder$:
                    \begin{equation}
                        \transformation_\blackHole :
                        \left\{
                        \begin{array}{ccc}
                             \vertices \cup \set{\blackHole} & \to     & \vertices \cup \set{\blackHole} \\
                             \vertex   & \mapsto & \blackHole
                        \end{array}
                        \right.
                        \;.
                        \label{transformationBlackHole}
                    \end{equation}
                    
                    Transformations do not take into consideration the edges of the graph.
                    To add the constraint that vertices should be mapped to vertices in their neighborhood, we introduce edge-constrained transformations:
                    
                    \begin{definition}[Edge-constrained (EC) transformation]
                        A transformation $\transformation \in \transformations{\graph}$ on a graph $\graph = \tuple{\vertices, \edges}$ is said to be \emph{edge-constrained} if:
                        $$
                            \forall \vertex \in \vertices_{\not\blackHole} : \edge{\vertex}{\transformation(\vertex)} \in \edges
                            \;.
                        $$
                        We denote the set of EC transformations on $\graph$ by $\EC{\graph}$, and the set of lossless EC transformations on $\graph$ by $\losslessEC{\graph}$.
                        \label{defEC}
                    \end{definition}
                    
                    \begin{proposition}
                        An EC transformation $\transformation \in \EC{\graph}$ on a graph $\graph = \tuple{\vertices, \edges}$ injectively defines an orientation of $\graph$, with the possible exclusion of vertices of which image is $\blackHole$.
                        \label{orientationOfEdges}
                    \end{proposition}
                    
                    \begin{proof}
                        Let $\adjacency$ be the adjacency matrix of $\graph$, and let $\adjacency_\transformation$ be a $\graphOrder \times \graphOrder$ matrix defined as follows:
                        $$
                            \forall \vertex_1, \vertex_2 \in \vertices :
                            \adjacency_\transformation\entry{\vertex_1, \vertex_2} =
                            \left\{
                                \begin{array}{cl}
                                    1 & \text{if~} \vertex_2 = \transformation(\vertex_1) \\
                                    0 & \text{otherwise}
                                \end{array}
                            \right.
                            \;.
                        $$
                        Remind that $\adjacency$ and $\adjacency_\transformation$ take their values in $\set{0, 1}$.
                        First, we show that $\forall \vertex_1, \vertex_2 \in \vertices : \adjacency_\transformation\entry{\vertex_1, \vertex_2} \leq \adjacency\entry{\vertex_1, \vertex_2}$.
                        Let us consider a vertex $\vertex_1 \in \vertices$.
                        There are three possible cases:
                        \begin{enumerate}
                            \item $\transformation(\vertex_1) = \blackHole$.
                                  In that case, $\forall \vertex_2 \in \vertices : \adjacency_\transformation\entry{\vertex_1, \vertex_2} = 0$.
                            \item $\exists \vertex_2 \in \neighborhood{1}{\vertex_1} : (\vertex_2 = \transformation(\vertex_1)) \wedge (\vertex_1 = \transformation(\vertex_2))$.
                                  In that case, $\adjacency_\transformation\entry{\vertex_1, \vertex_2} = \adjacency_\transformation\entry{\vertex_2, \vertex_1} = \adjacency\entry{\vertex_1, \vertex_2} = 1$, and $\forall \vertex_3 \in \vertices, \vertex_3 \neq \vertex_2 : \adjacency_\transformation\entry{\vertex_1, \vertex_3} = \adjacency_\transformation\entry{\vertex_3, \vertex_1} = 0$ (due to the injectivity of $\transformation$).
                            \item $\exists \vertex_2 \in \neighborhood{1}{\vertex_1} : (\vertex_2 = \transformation(\vertex_1)) \wedge (\vertex_1 \neq \transformation(\vertex_2))$.
                                  In that case, $\adjacency_\transformation\entry{\vertex_1, \vertex_2} = \adjacency\entry{\vertex_1, \vertex_2} = 1$, and $\adjacency_\transformation\entry{\vertex_2, \vertex_1} < \adjacency\entry{\vertex_2, \vertex_1}$, and $\forall \vertex_3 \in \vertices, \vertex_3 \neq \vertex_2 : \adjacency_\transformation\entry{\vertex_1, \vertex_3} = \adjacency_\transformation\entry{\vertex_3, \vertex_1} = 0$ (due to the injectivity of $\transformation$).
                        \end{enumerate}
                        In all cases, entries of $\adjacency_\transformation$ are lower or equal than those of $\adjacency$.
                        Additionally, due to case 3), there may exist $\vertex_1, \vertex_2 \in \vertices : \adjacency_\transformation\entry{\vertex_1, \vertex_2} < \adjacency\entry{\vertex_1, \vertex_2}$.
                        Therefore, $\adjacency_\transformation$ is not necessarily symmetric, and corresponds to a digraph in which every diedge contains elements that form an edge in $\edges$.

                        Additionally, if $\adjacency_{\transformation_1} = \adjacency_{\transformation_2}$, then $\forall \vertex \in \vertices : \transformation_1(\vertex) = \transformation_2(\vertex)$, \ie, $\transformation_1 = \transformation_2$.
                        So the mapping $\transformation \mapsto \adjacency_\transformation$ is injective.
                    \end{proof}
                    
                    The proof of \propref{orientationOfEdges} shows that it is possible to represent an EC transformation $\transformation$ on a graph $\graph = \tuple{ \vertices, \edges}$ by a digraph $\digraph{\transformation} = \tuple{\vertices, \diedges{\transformation}}$, with
                    $$
                        \forall \vertex \in \vertices : (\transformation(\vertex) \neq \blackHole) \Leftrightarrow \left(\diedge{\vertex}{\transformation(\vertex)} \in \diedges{\transformation}\right)
                        \;.
                    $$
                    
                    This allows a visual representation of EC transformations on a graph, where edges of $\edges$ are depicted with dotted lines, on top of which edges of $\diedges{\transformation}$ are drawn with plain arrows.
                    Additionally, we mark the vertices that have their image being $\blackHole$ by coloring them in black.
                    \figref{petersenTransformations} depicts an EC transformation on an example graph.
                    
                    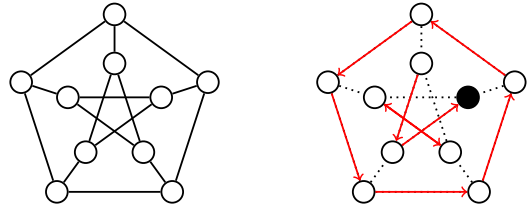
\begin{figure}
                        \centering
                        \scalebox{0.87}
                        {
                            \begin{tikzpicture}[thick]
                              \foreach \i/\angle in {1/18,2/90,3/162,4/234,5/306}{
                                \node (out\i) [draw, circle] at (\angle:1.5cm) {};
                                \node (in\i) [draw, circle] at (\angle:0.75cm) {};
                              }
                              \path
                              \foreach \i in {1,...,5}{
                                (out\i) edge (in\i)        
                              }
                              (out1) edge (out2)
                              (out2) edge (out3)
                              (out3) edge (out4)
                              (out4) edge (out5)
                              (out5) edge (out1)
                              (in1) edge (in3)
                              (in3) edge (in5)
                              (in5) edge (in2)
                              (in2) edge (in4)
                              (in4) edge (in1);
                            \end{tikzpicture}
                            ~~~~~~~~~~
                            \begin{tikzpicture}[thick]
                              \foreach \i/\angle in {1/18,2/90,3/162,4/234,5/306}{
                                \node (out\i) [draw, circle] at (\angle:1.5cm) {};
                                \node (in\i) [draw, circle] at (\angle:0.75cm) {};
                              }
                              \node (in1) [fill=black, circle] at (18:0.75cm) {};
                              \path[dotted]
                              \foreach \i in {1,...,5}{
                                (out\i) edge (in\i)        
                              }
                              (out1) edge (out2)
                              (out2) edge (out3)
                              (out3) edge (out4)
                              (out4) edge (out5)
                              (out5) edge (out1)
                              (in1) edge (in3)
                              (in3) edge (in5)
                              (in5) edge (in2)
                              (in2) edge (in4)
                              (in4) edge (in1);
                              \path[->, red]
                              (out1) edge (out2)
                              (out2) edge (out3)
                              (out3) edge (out4)
                              (out4) edge (out5)
                              (out5) edge (out1)
                              (in2) edge (in4)
                              (in4) edge (in1)
                              (in3) edge (in5)
                              (in5) edge (in3);
                            \end{tikzpicture}
                        }
                        \caption{Example of a simple, undirected graph (left) and an associated EC transformation with loss $1$ (right). The vertex whose image is $\blackHole$ is filled in black.}
                        \label{petersenTransformations}
                    \end{figure}
                    
                    Using this correspondence with a digraph, we can reformulate the loss of an EC transformation as follows:
                    
                    \begin{proposition}
                        Let $\transformation \in \EC{\graph}$ be an EC transformation on a graph $\graph = \tuple{\vertices, \edges}$, with associated digraph $\digraph{\transformation} = \tuple{\vertices, \diedges{\transformation}}$.
                        Let $\adjacency_\transformation$ be the adjacency matrix associated with $\digraph{\transformation}$, then:
                        $$
                            \loss(\transformation) = \graphOrder - \cardinal{\condSet{\set{\vertex_1, \vertex_2} \in \binom{\vertices}{2}}{\adjacency_\transformation\entry{\vertex_1, \vertex_2} = 1}}
                            \;.
                        $$
                        \label{lossMatrix}
                    \end{proposition}
                    
                    \begin{proof}
                        Let $\vertex_1 \in \vertices$.
                        If $\transformation(\vertex_1) \in \vertices$, then due to injectivity, there is a unique $\vertex_2 \in \vertices$ such that $\adjacency_\transformation\entry{\vertex_1, \vertex_2} = 1$.
                        If $\transformation(\vertex_1) = \blackHole$, then $\forall \vertex_2 \in \vertices : \adjacency_\transformation\entry{\vertex_1, \vertex_2} = 0$.
                        Therefore, $\loss(\transformation)$ is $\graphOrder$ minus the number of vertices that have an image in $\vertices$. 
                    \end{proof}
                    
                    Note that not all graphs admit lossless EC transformations.
                    Indeed, we can derive a few sufficient properties as well as necessary ones for an EC transformation to be lossless:
                    
                    \begin{proposition}
                        Consider a graph $\graph = \tuple{\vertices, \edges}$.
                        In order to have $\losslessEC{\graph} \neq \emptyset$, we have the following properties:
                        \begin{enumerate}
                            \item (Necessary): $\forall \vertex \in \vertices : \cardinal{\neighborhood{1}{\vertex}} > 0$;
                            \item (Necessary): No vertex is the unique neighbor for two other vertices;
                            \item (Sufficient): There exists an Hamiltonian cycle in $\graph$, \ie, a cycle that contains every vertex of $\vertices$ exactly once;
                            \item (Sufficient): There exists a perfect matching between all vertices in $\vertices$, \ie, there is a subset $\edges'$ of $\edges$ such that every vertex appears exactly once in the edges of $\edges'$.
                        \end{enumerate}
                        \label{conditionsECTransformations}
                    \end{proposition}
    
                    \begin{proof}
                        Let $\transformation \in \EC{\graph}$ be an EC transformation.
                        Let us consider the properties in the same order as above:
                        \begin{enumerate}
                            \item Let $\vertex \in \vertices$.
                                  If $\cardinal{\neighborhood{1}{\vertex}} = 0$, then the case $\edge{\vertex}{\transformation(\vertex)} \in \edges$ of \defref{defEC} is never matched, therefore $\transformation(\vertex) = \blackHole$.
                            \item Let $\vertex_1, \vertex_2, \vertex_3 \in \vertices$, with $\neighborhood{1}{\vertex_1} = \set{\vertex_3}$ and $\neighborhood{1}{\vertex_2} = \set{\vertex_3}$.
                                  To avoid the case where a vertex has its image equal to $\blackHole$, we must have $\transformation(\vertex_1) = \vertex_3$ and $\transformation(\vertex_2) = \vertex_3$.
                                  However, this contradicts injectivity of transformations.
                            \item Let $\vertex_1 \to \vertex_2 \to \dots \to \vertex_\graphOrder \to \vertex_1$ be a Hamiltonian cycle.
                                  The transformation that associates every vertex with its sucessor in the cycle is EC, and lossless.
                            \item If a perfect matching exists, we can determine $\edges' \subset \edges$ with $\cardinal{\edges'} = \frac{\graphOrder}{2}$ such that $\forall \vertex_1 \in \vertices : \exists \vertex_2 \in \vertices : \edge{\vertex_1}{\vertex_2} \in \edges'$.
                                  In this case, the transformation that associates with $\vertex_1$ its neighbor $\vertex_2$ is EC and lossless.
                        \end{enumerate}
                    \end{proof}
                
                    Among all transformations, we are in particular interested in translations.
                    Since their definition is not straightforward, let us first introduce the following properties for transformations:
                    
                    \begin{definition}[Weakly neighborhood-preserving (WNP) transformation]
                        A transformation $\transformation \in \transformations{\graph}$ on a graph $\graph = \tuple{\vertices, \edges}$ is said to be \emph{weakly neighborhood-preserving} if:
                        \begin{equation*}
                            \begin{split}
                                \forall \vertex_1, \vertex_2 \in \vertices_{\not\blackHole} : \edge{\vertex_1}{\vertex_2} \in \edges \Rightarrow \edge{\transformation(\vertex_1)}{\transformation(\vertex_2)} \in \edges
                                \;.
                            \end{split}
                        \end{equation*}
                        We note $\WNP{\graph}$ the set of WNP transformations on $\graph$, and $\losslessWNP{\graph}$ the set of lossless WNP transformations on $\graph$.
                    \end{definition}
                    
                    Informally, WNP transformations conserve existing neighborhoods.
                    However, note that two vertices that are not neighbors may be associated with neighboring vertices through a WNP transformation.
                    Transformations that do not create additional neighborhoods are characterized as follows:
                    
                    \begin{definition}[Strongly neighborhood-preserving (SNP) transformation]
                        A transformation $\transformation \in \transformations{\graph}$ on a graph $\graph = \tuple{\vertices, \edges}$ is said to be \emph{strongly neighborhood-preserving} if:
                        \begin{equation*}
                            \forall \vertex_1, \vertex_2 \in \vertices_{\not\blackHole} : \edge{\vertex_1}{\vertex_2} \in \edges \Leftrightarrow \edge{\transformation(\vertex_1)}{\transformation(\vertex_2)} \in \edges
                            \;.
                        \end{equation*}
                        We denote the set of SNP transformations on $\graph$ by $\SNP{\graph}$, and the set of lossless SNP transformations on $\graph$ by $\losslessSNP{\graph}$.
                        \label{snpDef}
                    \end{definition}

                    We can now define translations on graphs as follows:
                    
                    \begin{definition}[Translation on a graph]
                        A \emph{translation} $\translation \in \transformations{\graph}$ on a graph $\graph = \tuple{\vertices, \edges}$ is an EC and SNP transformation.
                        We denote the set of translations on $\graph$ by $\translations{\graph}$, and the set of lossless translations on $\graph$ by $\losslessTranslations{\graph}$.
                        \label{translationDef}
                    \end{definition}
                    
                    \begin{figure}
                        \centering
                        \scalebox{0.87}
                        {
                            \begin{tikzpicture}[thick]
                              \foreach \i/\angle in {1/18,2/90,3/162,4/234,5/306}{
                                \node (out\i) [draw, circle] at (\angle:1.5cm) {};
                                \node (in\i) [draw, circle] at (\angle:0.75cm) {};
                              }
                              \node (out1) [fill=black, circle] at (18:1.5cm) {};
                              \node (out2) [fill=black, circle] at (90:1.5cm) {};
                              \node (out3) [fill=black, circle] at (162:1.5cm) {};
                              \node (out4) [fill=black, circle] at (234:1.5cm) {};
                              \node (in3) [fill=black, circle] at (162:0.75cm) {};
                              \path[dotted]
                              \foreach \i in {1,...,5}{
                                (out\i) edge (in\i)        
                              }
                              (out1) edge (out2)
                              (out2) edge (out3)
                              (out3) edge (out4)
                              (out4) edge (out5)
                              (out5) edge (out1)
                              (in1) edge (in3)
                              (in3) edge (in5)
                              (in5) edge (in2)
                              (in2) edge (in4)
                              (in4) edge (in1);
                              \path[->, red]
                              (out5) edge (in5)
                              (in5) edge (in2)
                              (in2) edge (in4)
                              (in4) edge (in1)
                              (in1) edge (in3);
                            \end{tikzpicture}
                            ~~~~~
                            \begin{tikzpicture}[thick]
                              \foreach \i/\angle in {1/18,2/90,3/162,4/234,5/306}{
                                \node (out\i) [draw, circle] at (\angle:1.5cm) {};
                                \node (in\i) [fill=black, circle] at (\angle:0.75cm) {};
                              }
                              \path[dotted]
                              \foreach \i in {1,...,5}{
                                (out\i) edge (in\i)        
                              }
                              (out1) edge (out2)
                              (out2) edge (out3)
                              (out3) edge (out4)
                              (out4) edge (out5)
                              (out5) edge (out1)
                              (in1) edge (in3)
                              (in3) edge (in5)
                              (in5) edge (in2)
                              (in2) edge (in4)
                              (in4) edge (in1);
                              \path[->, red]
                              (out1) edge (out2)
                              (out2) edge (out3)
                              (out3) edge (out4)
                              (out4) edge (out5)
                              (out5) edge (out1);
                            \end{tikzpicture}
                        }
                      \caption{Examples of transformations that are translations on the Petersen graph in \figref{petersenTransformations}.}
                      \label{petersenTranslations}
                    \end{figure}
                    
                    \figref{petersenTranslations} depicts two examples of translations on a graph.
                    Again, note that the function $\transformation_\blackHole$ introduced in \eqref{transformationBlackHole} is a translation for any graph $\graph = \tuple{\vertices, \edges}$.
                    Additionally, we observe the following property:
                    
                    \begin{proposition}
                        Let $\translation \in \translations{\graph}$ be a translation on a graph $\graph = \tuple{\vertices, \edges}$, with associated digraph $\digraph{\translation} = \tuple{\vertices, \diedges{\translation}}$.
                        Edges of $\diedges{\translation}$ can be partitioned into directed cycles, and directed paths that have one vertex for which the image is $\blackHole$.
                        \label{translationsAreCyclesOrPaths}
                    \end{proposition}
                    
                    \begin{proof}
                        By injectivity of transformations, any vertex $\vertex_1 \in \vertices$ has an image by $\translation$ which is either a vertex $\vertex_2 \in \vertices$ with no other inverse image, or $\blackHole$.
                        Therefore, every vertex belongs either to a path $\vertex_1 \to \vertex_2 \to \dots \to \blackHole$, or a cycle $\vertex_1 \to \vertex_2 \to \dots \to \vertex_1$.
                        Additionally, the associated digraph restricts the existence of these paths and cycles to paths and cycles that exist in $\edges$.
                    \end{proof}
                    
                    It is also interesting to notice that every translation admits an inverse translation with the same loss:
                    
                    \begin{proposition}
                        Let $\translation \in \translations{\graph}$ be a translation on a graph $\graph = \tuple{\vertices, \edges}$, with associated digraph $\digraph{\translation} = \tuple{\vertices, \diedges{\translation}}$ of adjacency matrix $\adjacency_\translation$.
                        Let us call $\translation^{-1}$ the inverse translation associated with the digraph $\digraph{\translation^{-1}} = \tuple{\vertices, \diedges{\translation^{-1}}}$, with $\diedge{\vertex_1}{\vertex_2} \in \diedges{\translation^{-1}} \Leftrightarrow \diedge{\vertex_2}{\vertex_1} \in \diedges{\translation}$.
                        We have the relation $\loss(\translation) = \loss(\translation^{-1})$.
                        \label{sameLossForInverse}
                    \end{proposition}
                    
                    \begin{proof}
                        First, let us notice that $\diedges{\translation^{-1}}$ is the exact same set of edges as in $\diedges{\translation}$, but with reverse direction.
                        Therefore, since $\translation$ is EC, it is also the case for $\translation^{-1}$.
                        Additionally, since $\translation$ is SNP, it preserves the existing neighborhoods and does not create additional ones.
                        Therefore, it is also the case for the converse, and $\translation^{-1}$ is thus SNP.
                        From \defref{translationDef}, $\translation^{-1}$ is therefore a translation.
                        Finally, from \propref{lossMatrix}, and noticing that the adjacency matrix associated with $\digraph{\translation^{-1}}$ is the transpose of $\adjacency_\translation$, we conclude.
                    \end{proof}
                    
                    Translations can be given a well-founded relation $\prec$:
                    \begin{equation*}
                        \begin{split}
                            \forall \translation_1, \translation_2 \in \translations{\graph} : (\translation_1 \prec \translation_2) \Leftrightarrow (\loss(\translation_1) > \loss(\translation_2)) \\
                            \wedge (\exists \vertex \in \vertices : \translation_1(\vertex) = \translation_2(\vertex))
                            \;.
                        \end{split}
                    \end{equation*}
                    
                    \begin{proposition}
                        Let $\translation_1, \translation_2, \translation_3 \in \translations{\graph}$. The relation $\prec$ has the following properties:
                        \begin{enumerate}
                            \item It is \emph{irreflexive}, \ie, $\translation_1 \prec \translation_1$ is not true.
                            \item It is \emph{antisymmetric}, \ie, it is not possible to have both $\translation_1 \prec \translation_2$ and $\translation_2 \prec \translation_1$.
                            \item It is \emph{intransitive}, \ie, it is not true that $((\translation_1 \prec \translation_2) \wedge (\translation_2 \prec \translation_3)) \Rightarrow (\translation_1 \prec \translation_3)$.
                        \end{enumerate}
                    \end{proposition}
                    
                    \begin{proof}
                        Let us consider the three properties separately:
                        \begin{enumerate}
                            \item By definition, $\translation_1$ is comparable to itself, since there exists at least one edge in common.
                                  Comparison is then made using $<$, which is an irreflexive order on $\integers$.
                            \item In the case where $\nexists \vertex \in \vertices : \translation_1(\vertex) = \translation_2(\vertex)$, then $\translation_1$ and $\translation_2$ are not comparable (noted $\translation_1 \sim \translation_2$).
                                  In the case where such an edge exists, $<$ is an antisymmetric order on $\integers$.
                            \item Let us consider the following graph: 
                                  $$
                                      \begin{tikzpicture}[thick]
                                          \node (0) [draw, circle] at (0,0) {};
                                          \node (1) [draw, circle] at (1,0) {};
                                          \node (2) [draw, circle] at (2,0) {};
                                          \node (3) [draw, circle] at (3,0) {};
                                          \node (4) [draw, circle] at (4,0) {};
                                          \node (5) [draw, circle] at (5,0) {};
                                          \path[]
                                          (0) edge (1)
                                          (1) edge (2)
                                          (2) edge (3)
                                          (3) edge (4)
                                          (4) edge (5);
                                    \end{tikzpicture}
                                  $$
                                  Let $\translation_1, \translation_2, \translation_3$ be the following translations:
                                  $$
                                      \begin{tikzpicture}[thick]
                                          \node (name) at (-1,0) {$\translation_1:$};
                                          \node (0) [draw, circle] at (0,0) {};
                                          \node (1) [fill=black, circle] at (1,0) {};
                                          \node (2) [fill=black, circle] at (2,0) {};
                                          \node (3) [fill=black, circle] at (3,0) {};
                                          \node (4) [fill=black, circle] at (4,0) {};
                                          \node (5) [fill=black, circle] at (5,0) {};
                                          \path[dotted]
                                          (0) edge (1)
                                          (1) edge (2)
                                          (2) edge (3)
                                          (3) edge (4)
                                          (4) edge (5);
                                          \path[->, red]
                                          (0) edge (1);
                                    \end{tikzpicture}
                                    \hspace{1.25cm}
                                  $$
                                  $$
                                      \begin{tikzpicture}[thick]
                                          \node (name) at (-1,0) {$\translation_2:$};
                                          \node (0) [draw, circle] at (0,0) {};
                                          \node (1) [draw, circle] at (1,0) {};
                                          \node (2) [fill=black, circle] at (2,0) {};
                                          \node (3) [fill=black, circle] at (3,0) {};
                                          \node (4) [fill=black, circle] at (4,0) {};
                                          \node (5) [draw, circle] at (5,0) {};
                                          \path[dotted]
                                          (0) edge (1)
                                          (1) edge (2)
                                          (2) edge (3)
                                          (3) edge (4)
                                          (4) edge (5);
                                          \path[->, red]
                                          (0) edge (1)
                                          (1) edge (2)
                                          (5) edge (4);
                                    \end{tikzpicture}
                                    \hspace{1.25cm}
                                  $$
                                  $$
                                      \begin{tikzpicture}[thick]
                                          \node (name) at (-1,0) {$\translation_3:$};
                                          \node (0) [fill=black, circle] at (0,0) {};
                                          \node (1) [draw, circle] at (1,0) {};
                                          \node (2) [draw, circle] at (2,0) {};
                                          \node (3) [draw, circle] at (3,0) {};
                                          \node (4) [draw, circle] at (4,0) {};
                                          \node (5) [draw, circle] at (5,0) {};
                                          \path[dotted]
                                          (0) edge (1)
                                          (1) edge (2)
                                          (2) edge (3)
                                          (3) edge (4)
                                          (4) edge (5);
                                          \path[->, red]
                                          (1) edge (0)
                                          (2) edge (1)
                                          (3) edge (2)
                                          (4) edge (3)
                                          (5) edge (4);
                                    \end{tikzpicture}
                                    \hspace{1.25cm}
                                  $$
                                  In this example, $\translation_1 \prec \translation_2$ and $\translation_2 \prec \translation_3$.
                                  However, $\translation_1$ and $\translation_3$ have no edge in common, thus $\translation_1 \sim \translation_3$.
                                  Still, note that $\translation_3^{-1}$, the inverse translation of $\translation_3$, is comparable with $\translation_1$ and $\translation_2$.
                                  Therefore, $\prec$ is not an \emph{antitransitive} relation.
                        \end{enumerate}
                    \end{proof}
                    
                    Using this relation, we define minimal translations:
                    
                    \begin{definition}[Minimal translation]
                        A translation $\translation_1 \in \translations{\graph}$ is \emph{minimal} if there is no $\translation_2 \in \translations{\graph}$ such that $\translation_1 \prec \translation_2$, \ie, if it minimizes the loss.
                    \end{definition}
                    
                    Indeed, lossless translations $\losslessTranslation \in \losslessTranslations{\graph}$ are necessarily minimal.
                    Additionally, we define pseudo-minimal translations:
                    
                    \begin{definition}[Pseudo-minimal translation]
                        \emph{Pseudo-minimal} translations are defined inductively.
                        A translation $\translation_1 \in \translations{\graph}$ is pseudo-minimal if one of the following holds:
                        \begin{enumerate}
                            \item $\translation_1$ is minimal;
                            \item Any translation $\translation_2 \in \translations{\graph}$ such that $\translation_1 \prec \translation_2$ is not pseudo-minimal.
                        \end{enumerate}
                    \end{definition}
                    
                    \begin{proposition}
                        Any graph $\graph = \tuple{\vertices, \edges}$ admits at least one minimal translation $\translation \in \translations{\graph}$.
                        Also, $\translation^{-1}$ is minimal.
                        \label{minimalTranslationsExist}
                    \end{proposition}
                    
                    \begin{proof}
                        In order to show that a minimal translation exists, let us study the following cases:
                        \begin{enumerate}
                            \item In the case where $\cardinal{\edges} = 0$, the only possible transformation is $\transformation_\blackHole$ \eqref{transformationBlackHole}, which is thus minimal;
                            \item In the more general case, let us consider an edge $\edge{\vertex_1}{\vertex_2} \in \edges$.
                                  The function $\translation_1$ such that $\translation_1(\vertex_1) = \vertex_2$, and $\forall \vertex_3 \neq \vertex_1 : \translation_1(\vertex_3) = \blackHole$ is obviously a translation.
                                  Now, consider a maximal sequence $(\translation_\i)_\i$ of translations of which first element is $\translation_1$ and such that $\forall \i: \translation_\i \prec \translation_{\i+1}$.
                                  This sequence is necessarily finite, since $\loss(\translation_\i)$ decreases and $<$ is a well-founded order on $\integers$.
                                  By definition, the last element $\translation_\j$ of this sequence is a minimal translation.
                                  
                                  Now, let us show that $\translation_\j^{-1}$ --- the inverse translation as defined in \propref{sameLossForInverse} --- is also minimal.
                                  From \propref{sameLossForInverse}, we have $\loss(\translation_\j) = \loss(\translation_\j^{-1})$.
                                  Now, let us imagine that there exists $\translation_\k$ such that $\translation_\j^{-1} \prec \translation_\k$.
                                  Using the same reasoning as above, and noticing that $\translation_\j^{-1}$ and $\translation_\k$ share at least one edge, we obtain that $\translation_\j \prec \translation_\k^{-1}$.
                                  Since $\translation_\j$ is minimal, we reach a contradiction.
                                  As a consequence, both $\translation_\j$ and $\translation_\j^{-1}$ are minimal translations.
                                  It is interesting to notice that a special case occurs when $\translation_\j$ is a perfect matching between all vertices in $\vertices$.
                                  In this situation, we have $\translation_\j = \translation_\j^{-1}$.
                        \end{enumerate}
                    \end{proof}
                    
                    To sum up the various sets we introduced in this section, \figref{vennSets} presents the corresponding Venn diagram.

                    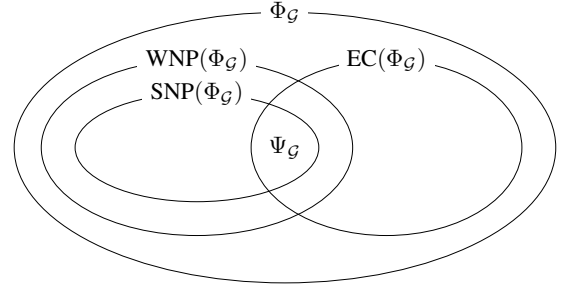
\begin{figure}
                        \centering
                        \scalebox{0.9}
                        {
                            \begin{tikzpicture}
                                \draw (0,0) ellipse (4cm and 2cm);
                                \draw (-1.3cm,0) ellipse (2.3cm and 1.3cm);
                                \draw (1.5cm,0) ellipse (2.0cm and 1.3cm);
                                \draw (-1.3cm,0) ellipse (1.8cm and 0.8cm);
                                \node (1) [fill=white] at (0, 2cm) {$\transformations{\graph}$};
                                \node (2) [fill=white] at (-1.3cm, 1.3cm) {$\WNP{\graph}$};
                                \node (3) [fill=white] at (-1.3cm, 0.8cm) {$\SNP{\graph}$};
                                \node (4) [fill=white] at (1.5cm, 1.3cm) {$\EC{\graph}$};
                                \node (5) [fill=white] at (0, 0) {$\translations{\graph}$};
                            \end{tikzpicture}
                        }
                        \caption{Venn diagram summarizing the various types of transformations introduced in this section.}
                        \label{vennSets}
                    \end{figure}

                \subsection{Isometries on graphs}
                    
                    Translations on Euclidean spaces are isometries.
                    When it comes to graphs, we also want the translations to be distance-preserving functions.
                    However, since in the general case there is no Euclidean space associated with the graph, we consider here the geodesic distance $\geodesic$.
                    
                    \begin{definition}[Isometry on a graph]
                        A transformation $\transformation \in \transformations{\graph}$ on a graph $\graph = \tuple{\vertices, \edges}$ is an \emph{isometry} if
                        \begin{equation*}
                            \begin{split}
                                \forall \vertex_1, \vertex_2 \in \vertices_{\not\blackHole} : \geodesic(\vertex_1, \vertex_2) = \geodesic(\transformation(\vertex_1), \transformation(\vertex_2))
                                \;.
                            \end{split}
                        \end{equation*}
                        We denote the set of isometries on $\graph$ by $\ISO{\graph}$, and the set of lossless isometries on $\graph$ by $\losslessISO{\graph}$.
                    \end{definition}
                    
                    Examples of isometries are the translations presented in \figref{petersenTranslations}.
                    
                    \begin{proposition}
                        Let $\graph = \tuple{\vertices, \edges}$ be a graph.
                        We have that $\losslessSNP{\graph} \subset \losslessISO{\graph}$.
                        \label{losslessSNPAreISO}
                    \end{proposition}
                    
                    \begin{proof}
                        Let $\vertex_1, \vertex_2 \in \vertices$.
                        For a lossless SNP transformation $\losslessTransformation \in \losslessSNP{\graph}$, we distinguish the following cases:
                        \begin{enumerate}
                            \item There is no path between vertices $\vertex_1$ and $\vertex_2$, \ie, $\geodesic(\vertex_1, \vertex_2)$ is infinite.
                                  This corresponds to the case where they belong to different connected components.
                                  By contradiction, let $\geodesic(\losslessTransformation(\vertex_1), \losslessTransformation(\vertex_2))$ be finite.
                                  This implies that there exists a path $\losslessTransformation(\vertex_1) \rightarrow \losslessTransformation(\vertex_{\i_1}) \rightarrow \losslessTransformation(\vertex_{\i_2}) \rightarrow \dots \rightarrow \losslessTransformation(\vertex_{\i_\k}) \rightarrow \losslessTransformation(\vertex_2)$ in the graph.
                                  Since $\losslessTransformation$ is SNP, we have that $\vertex_1 \rightarrow \vertex_{\i_1} \rightarrow \vertex_{\i_2} \rightarrow \dots \rightarrow \vertex_{\i_\k} \rightarrow \vertex_2$ is also a path in the graph, therefore we reach a contradiction.
                                  As a consequence, $\geodesic(\vertex_1, \vertex_2)$ and $\geodesic(\losslessTransformation(\vertex_1), \losslessTransformation(\vertex_2))$ are both infinite.
                            \item A shortest path $\vertex_1 \rightarrow \vertex_{\i_1} \rightarrow \vertex_{\i_2} \rightarrow \dots \rightarrow \vertex_{\i_\k} \rightarrow \vertex_2$ exists.
                                  Since $\losslessTransformation$ is lossless, $\losslessTransformation(\vertex_1) \rightarrow \losslessTransformation(\vertex_{\i_1}) \rightarrow \losslessTransformation(\vertex_{\i_2}) \rightarrow \dots \rightarrow \losslessTransformation(\vertex_{\i_\k}) \rightarrow \losslessTransformation(\vertex_2)$ is also a path (no intermediar vertex has its image equal to $\blackHole$).
                                  Additionally, $\losslessTransformation$ being SNP, it does not create nor remove neighborhoods, so $\losslessTransformation(\vertex_1) \rightarrow \losslessTransformation(\vertex_{\i_1}) \rightarrow \losslessTransformation(\vertex_{\i_2}) \rightarrow \dots \rightarrow \losslessTransformation(\vertex_{\i_\k}) \rightarrow \losslessTransformation(\vertex_2)$ is also a shortest path from $\losslessTransformation(\vertex_1)$ to $\losslessTransformation(\vertex_2)$.
                                  Therefore, we have $\geodesic(\vertex_1, \vertex_2) = \geodesic(\losslessTransformation(\vertex_1), \losslessTransformation(\vertex_2))$.
                        \end{enumerate}
                    \end{proof}
                
                    \begin{corollary}
                        Let $\graph = \tuple{\vertices, \edges}$ be a graph.
                        $\losslessTranslations{\graph} \subset \losslessISO{\graph}$.
                    \end{corollary}
                    
                    \begin{proof}
                        $\losslessTranslations{\graph} \subset \losslessSNP{\graph} \subset \losslessISO{\graph}$.
                    \end{proof}
                    
                    Note that \propref{losslessSNPAreISO} holds for lossless SNP transformations only.
                    In the more general case of SNP transformations $\transformation \in \SNP{\graph}$, having a vertex that has its image equal to $\blackHole$ may cause $\geodesic(\transformation(\vertex_1), \transformation(\vertex_2))$ and $\geodesic(\vertex_1, \vertex_2)$ to be different.
                    As an example, \figref{SNPAreNotISO} depicts a SNP transformation $\transformation \in \SNP{\graph}$ for which $\geodesic(\vertex_1, \vertex_2) < \geodesic(\transformation(\vertex_1), \transformation(\vertex_2))$.
                    When considering the inverse transformation $\transformation^{-1}$, we have $\geodesic(\vertex_1, \vertex_2) > \geodesic(\transformation^{-1}(\vertex_1), \transformation^{-1}(\vertex_2))$.
                    
                    \begin{figure}
                        \centering
                        \begin{tikzpicture}[thick]
                          \node (1) [draw, circle] at (1, 0) {};
                          \node (2) [draw, circle] at (2, 1) {};
                          \node (3) [draw, circle] at (2, 0) {$\vertex_1$};
                          \node (4) [draw, circle] at (3, 0) {};
                          \node (5) [fill=black, circle] at (4.5, 0) {};
                          \node (6) [draw, circle] at (6, 0) {};
                          \node (7) [draw, circle] at (7, 1) {};
                          \node (8) [draw, circle] at (7, 0) {$\vertex_2$};
                          \node (9) [draw, circle] at (8, 0) {};
                          \path[dotted]
                          (4) edge (5)
                          (5) edge (6);
                          \path[->, red]
                          (1) edge (2)
                          (3) edge (1)
                          (2) edge (4)
                          (4) edge (3)
                          (6) edge (8)
                          (8) edge (9)
                          (9) edge (7)
                          (7) edge (6);
                        \end{tikzpicture}
                        \caption
                        {
                            Example of a SNP transformation $\transformation \in \SNP{\graph}$ with a non-zero loss that is not an isometry.
                            In this example, $\geodesic(\vertex_1, \vertex_2) = 4$, $\geodesic(\transformation(\vertex_1), \transformation(\vertex_2)) = 6$ and $\geodesic(\transformation^{-1}(\vertex_1), \transformation^{-1}(\vertex_2)) = 2$.
                        }
                      \label{SNPAreNotISO}
                    \end{figure}
                    
                    Still, it is interesting to note that some transformations with a non-zero loss are isometries.
                    Examples of such are depicted in \figref{petersenTranslations}.

            \section{Results and proofs}
            \label{resultsProofs}
                
                In this section, we are interested in identifying translations on arbitrary graphs.
                After introducing some generic results, we study the particular case of the grid and torus graphs, and show that some particular translations of ours match the underlying Euclidean ones.
                
                \subsection{Results on generic graphs}
                    
                    As stated before, the function $\transformation_\blackHole$ introduced in \eqref{transformationBlackHole} is a translation for any graph.
                    However, it is not very interesting, since it destroys all signal information when translating it.
                    Therefore, we need to identify more complex translations that keep most of the signal entries.
                    As a consequence, we are particularly interested in minimal translations.

                    Before trying to identify translations, let us provide bounds on the number of translations:
                    
                    \begin{proposition}
                        A graph with order $\graphOrder$ cannot admit more than
                        $$
                            \sum_{\k = 0}^\graphOrder \frac{1}{(\graphOrder-\k)!} \sum_{\j = 0}^\k (-1)^\j \binom{\k}{\j} (\graphOrder - \j)!
                        $$
                        translations.
                        This number is reached for the complete graph $\graph_c = \tuple{\vertices_c, \edges_c}$.
                        \label{nbTranslations}
                    \end{proposition}
                    
                    \begin{proof}
                        Every EC transformation on $\graph_c$ is necessary SNP, since all vertices are pairwise linked and therefore share the same neighborhood.
                        However, not every transformation on such graph is EC, since transformations can map vertices to themselves, and we consider simple graphs only.
                        Therefore, for a fixed loss $\graphOrder - \k$ ($\k \leq \graphOrder$), the set of translations of loss $\graphOrder - \k$ is exactly the set of injective functions that have no fixed points of $\k$ elements to $\graphOrder$.
                        The cardinal of such a set is given by the solution of the \emph{$(\graphOrder, \k)$-matching problem} in \cite{10.2307/2689812} as follows:
                        \begin{equation}
                            \frac{1}{(\graphOrder-\k)!} \sum_{\j = 0}^\k (-1)^\j \binom{\k}{\j} (\graphOrder - \j)!
                            \;.
                            \label{nkMatching}
                        \end{equation}
                        By summing for every possible value of $\k$, corresponding to the number of vertices that have an image different to $\blackHole$, we obtain the number of translations on $\graph_c$.
                        Then, note that any graph $\graph = \tuple{\vertices, \edges}$ of order $\graphOrder$ has its edges $\edges \subset \edges_c$.
                        As a consequence, since any EC transformation is a translation on $\graph_c$, it follows that any translation on $\graph$ is also a translation on $\graph_c$.
                        Therefore, a graph of order $\graphOrder$ cannot admit more translations than the complete graph.
                    \end{proof}
                    
                    This characterization of the number of translations gives us the following result on the number of minimal translations:
                    
                    \begin{proposition}
                        A graph of order $\graphOrder$ cannot admit more than
                        $$
                            \graphOrder! \sum_{\j = 0}^\graphOrder \frac{(-1)^\j}{\j!}
                        $$
                        minimal translations.
                        This number is reached for the complete graph $\graph_c = \tuple{\vertices_c, \edges_c}$.
                        \label{nbMinTranslations}
                    \end{proposition}
                    
                    \begin{proof}
                        Minimal translations on $\graph_c$ are necessarily lossless, since any translation shares at least a diedge with an Hamiltonian cycle on this graph, which is lossless.
                        By particularizing \eqref{nkMatching} for $\k = \graphOrder$, we obtain the number of lossless translations on $\graph_c$, which is exactly the number of derangements of a set of $\graphOrder$ elements, \ie, the number of permutations of $\graphOrder$ elements with no fixed points \cite{de1713essay}.
                        
                        Using the same reasoning as in the proof of \propref{nbTranslations}, any minimal translation on a graph $\graph$ is included in a minimal translation on $\graph_c$.
                        Therefore, a graph of order $\graphOrder$ cannot admit more minimal translations than the complete graph.
                    \end{proof}
                    
                    The number of translations on graphs is therefore exponential in the general case.
                    Additionally, we prove that identifying translations is a complex problem:
                    
                    \begin{proposition}
                        The problem of deciding, for an input graph $\graph = \tuple{\vertices, \edges}$ and two subsets $\vertices_1$ and $\vertices_2$ of $\vertices$, if there is a translation for which the image is exactly $\vertices_2$ and with inverse images only in $\vertices_1$ is NP-complete.
                        \label{npCompletenessIdentification}
                    \end{proposition}
                    
                    \begin{proof}
                        To prove this result, we first prove that the problem is NP, and then that it is NP-hard.
                        
                        All possible transformations with inverse image set $\vertices_1$ and image set $\vertices_2$ can be generated non-deterministically by induction as follows:
                        \begin{enumerate}
                            \item $\transformation_\blackHole \in \transformations{\graph}$;
                            \item If we have a transformation $\transformation_1 \in \transformations{\graph}$, and two vertices $\vertex_1 \in \vertices_1$, $\vertex_2 \in \vertices_2$ such that $\transformation_1(\vertex_1) = \blackHole$ and  $\nexists \vertex_3 \in \vertices_1 : \transformation_1(\vertex_3) = \vertex_2$, then define $\transformation_2 \in \transformations{\graph}$ as follows:
                                  \begin{itemize}
                                      \item $\transformation_2(\vertex_1) = \vertex_2$;
                                      \item $\forall \vertex_3 \in \vertices_1, \vertex_3 \neq \vertex_1 : \transformation_2(\vertex_3) = \transformation_1(\vertex_3)$.
                                  \end{itemize}
                        \end{enumerate}
                        Furthermore, determining whether any such transformation is a translation or not can be done by checking EC and SNP constraints, which can be done in polynomial time.
                        So the problem is NP.
                    
                        Then, we prove the problem is NP-hard by reduction from the subgraph isomorphism problem.
                        Consider two graphs $\graph_1 = \tuple{\vertices_1, \edges_1}$ and $\graph_2 = \tuple{\vertices_2, \edges_2}$.
                        Without loss of generality, we consider that $\vertices_1 \cap \vertices_2 = \emptyset$.
                        The more general case where $\vertices_1 \cap \vertices_2 \neq \emptyset$ can be included in the proof by duplication of the vertices in the intersection.
                        
                        From these graphs, we build the graph $\graph_3 = \tuple{\vertices_1 \cup \vertices_2, \edges_3}$, where $\edges_3 = \edges_1 \cup \edges_2 \cup \set{\edge{\vertex_1}{\vertex_2}, \vertex_1 \in \vertices_1, \vertex_2 \in \vertices_2}$.
                        Note that this construction is at most quadratic in the order of $\graph$.
                        Then, we show that answering our problem on $\graph_3$ solves the problem of subgraph isomorphism between $\graph_1$ and $\graph_2$.
                        
                        To this end, consider the following two properties, that we prove to be equivalent:
                        \begin{enumerate}
                            \item There is a translation of which image set is $\vertices_2$ and inverse images are in $\vertices_1$;
                            \item There is a subgraph of $\graph_1$ isomorph to $\graph_2$.
                        \end{enumerate}
                        
                        We prove this in two steps.
                        First, consider there exists such a translation.
                        Then, since it is SNP, the subgraph corresponding to the inverse images of vertices in $\vertices_2$ is isomorph to $\graph_2$.
                        
                        Conversely, consider there exists an isomorphism, then the transformation that associates each vertex in $\vertices_2$ with its corresponding vertex in $\vertices_1$ is a translation.
                        indeed, it is EC because of the complete bipartite subgraph connecting vertices in $\vertices_2$ to vertices in $\vertices_1$, and it is SNP as a particularization of the isomorphism property.
                        
                        As a consequence, the problem of deciding, for an input graph $\graph = \tuple{\vertices, \edges}$ and two subsets $\vertices_1$ and $\vertices_2$ of $\vertices$, if there is a translation for which the image is exactly $\vertices_2$ and with inverse images only in $\vertices_1$ is at least as difficult as the subgraph isomorphism problem.
                        Since it is also NP, it is NP-complete.
                    \end{proof}
                    
                    These results tell us that finding the translations on a graph is a hard problem.
                    In \secref{relaxation}, we introduce relaxations of the proposed translations to tackle this issue.

                \subsection{Correspondence with Euclidean translations on metric spaces}
                
                    In the following paragraphs, we focus on the particular case of highly regular graphs, namely the torus graph and the grid graph.
                    We show that, without taking into account the underlying metric space, the translations that minimize the loss correspond to Euclidean translations of images.
                    
	                \subsubsection{Application to the torus graph}
	                \label{resultsTorus} ~
	                
	                    Let us first consider the case of the torus graph $\graph_t = \tuple{\vertices_t, \edges_t}$ yielded by a dimensions vector $\gridDimensions \in \strictlyPositiveIntegers^\gridNbDimensions$.
	                    Such graphs are highly regular, and are often used to model classical domains, such as the periodical time with a $1$-dimensional torus graph, or the pixels of a periodical image with a $2$-dimensional torus graph.
	                    Additionally, what makes these graphs interesting is the fact that they are constructed using an Euclidean space (see \defref{defTorus}).
	                    For this reason, we use the notation $\vertex \in \vertices_t$ when referring to the index of vertex $\vertex$, and $\vertexVector \in \vertices_t$ when referring to its Euclidean coordinates in $\gridDimensions$.
	                    
	                    In this section, we aim to find a relation between translations defined on an Euclidean space, and those defined on the graph, with no reference to the underlying metrics.
	                    To do so, let us first formalize the notion of Euclidean translation on $\graph_t$.
	                    
	                    \begin{definition}[Euclidean translation on the torus graph]
	                        An \emph{Euclidean translation} $\translation_t$ on the torus graph is such that:
	                        $$
	                            \exists \differenceVector \in \positiveIntegers^\gridNbDimensions : \forall \vertexVector \in \vertices_t : \translation_t(\vertexVector) = \vertexVector + \differenceVector
	                            \;.
	                        $$
	                        \label{defEuclideanTranslationTorus}
	                    \end{definition}
	                    
	                    \begin{definition}[Dirac vector for dimension $i$]
	                        Let $\diracVector_i$ be the vector in $\set{0, 1}^\gridNbDimensions$ with a single non-null entry $\i \in \intInterval{1, \gridNbDimensions}$, defined as follows:
	                        $$
	                            \forall \j \in \intInterval{1, \gridNbDimensions} :
	                            \diracVector_\i\entry{\j} = \left\{
	                                \begin{array}{cl}
	                                    1 & \text{if~} \j = \i \\
	                                    0 & \text{otherwise}
	                                \end{array}
	                            \right.
	                            \;.
	                        $$
	                    \end{definition}
	                    Remember from \defref{defTorus} that coordinates of the torus graph are defined in $\integersModulo{\gridDimensions\entry{\i}}$ (for the $\i\th$ coordinate).
	                    Therefore, addition and subtraction take into account the modulo.
	                    By construction of the torus graph, we have that $\forall \vertexVector \in \vertices_t, \forall \i \in \intInterval{1, \gridNbDimensions} : \edge{\vertexVector}{\vertexVector + \diracVector_\i} \in \edges_t$.
	                    
	                    Note that this is also true for the inverse Dirac vectors containing a single non-null entry $\i$ being $-1$.
	                    Consequently, the following results also apply using such vectors.
	                    
	                    \begin{lemma}[Contamination lemma on the torus graph]
	                        Let $\translation \in \losslessTranslations{\graph_t}$ be a lossless translation on the torus graph, with $\forall \i \in \intInterval{1, \gridNbDimensions} : \gridDimensions\entry{\i} \geq 5$.
	                        Let $\vertexVector_1 \in \vertices_t$.
	                        Let us consider the Dirac vector $\diracVector_\j = \translation(\vertexVector_1) - \vertexVector_1$.
	                        Then, $\forall \vertexVector_2 \in \vertices_t : \translation(\vertexVector_2) = \vertexVector_2 + \diracVector_\j$.
	                        \label{contaminationlemmaTorus}
	                    \end{lemma}
	                    
	                    \begin{proof}
	                        We proceed in two steps:
	                        \begin{enumerate}
	                            \item First, let us show that $\translation(\vertexVector_1 - \diracVector_\j) = \vertexVector_1$.
	                                  By construction of the torus graph in \defref{defTorus}, we have that $\edge{\vertexVector_1 - \diracVector_\j}{\vertexVector_1} \in \edges_t$.
	                                  Since $\translation$ is EC, we must have $\translation(\vertexVector_1 - \diracVector_\j) \in \neighborhood{1}{\vertexVector_1 - \diracVector_\j}$.
	                                  Also, since $\translation$ is SNP, we must have $\translation(\vertexVector_1 - \diracVector_\j) \in \neighborhood{1}{\translation(\vertexVector_1)}$.
	                                  As a consequence, $\translation(\vertexVector_1 - \diracVector_\j) \in \neighborhood{1}{\vertexVector_1 - \diracVector_\j} \cap \neighborhood{1}{\translation(\vertexVector_1)}$.
	                                  The neighborhood of $\translation(\vertexVector_1)$ is $\neighborhood{1}{\translation(\vertexVector_1)} = \neighborhood{1}{\vertexVector_1 + \diracVector_\j} = \{\vertexVector_1 + \diracVector_\j + \diracVector_1, \vertexVector_1 + \diracVector_\j - \diracVector_1, \vertexVector_1 + \diracVector_\j + \diracVector_2, \vertexVector_1 + \diracVector_\j - \diracVector_2, \dots, \vertexVector_1 + 2 \diracVector_\j, \vertexVector_1, \dots, \vertexVector_1 + \diracVector_\j + \diracVector_\gridNbDimensions, \vertexVector_1 + \diracVector_\j - \diracVector_\gridNbDimensions\}$.
	                                  Similarly, the neighborhood of $\vertexVector_1 - \diracVector_\j$ is $\neighborhood{1}{\vertexVector_1 - \diracVector_\j} = \{\vertexVector_1 - \diracVector_\j + \diracVector_1, \vertexVector_1 - \diracVector_\j - \diracVector_1, \vertexVector_1 - \diracVector_\j + \diracVector_2, \vertexVector_1 - \diracVector_\j - \diracVector_2, \dots, \vertexVector_1, \vertexVector_1 - 2 \diracVector_\j, \dots, \vertexVector_1 - \diracVector_\j + \diracVector_\gridNbDimensions, \vertexVector_1 - \diracVector_\j - \diracVector_\gridNbDimensions\}$.
	                                  Since $\forall \i \in \intInterval{1, \gridNbDimensions} : \gridDimensions\entry{\i} \geq 3$, we have $\vertexVector_1 + \diracVector_\j + \diracVector_\k \neq \vertexVector_1 - \diracVector_\j + \diracVector_\k$ ($\j \neq \k$).
	                                  Therefore, vertices with coordinates that differ by an entry in dimension $\k \neq \j$ cannot belong to the intersection by construction of the torus graph.
	                                  Similarly, because $\forall \i \in \intInterval{1, \gridNbDimensions} : \gridDimensions\entry{\i} \geq 5$, we obtain that $\vertexVector_1 + 2 \diracVector_\j$ cannot be the same vertex as $\vertexVector_1 - 2 \diracVector_\j$, since they differ by $\alpha \diracVector_\j$ ($\alpha > 1$).
	                                  As a consequence, $\neighborhood{1}{\vertexVector_1 - \diracVector_\j} \cap \neighborhood{1}{\translation(\vertexVector_1)} = \set{\vertexVector_1}$, and thus $\translation(\vertexVector_1 - \diracVector_\j) = \vertexVector_1$.
	                            
	                            \item Now, let us consider a vertex $\vertexVector_2 \in \neighborhood{1}{\vertex_1} \setminus \set{\vertex_1 - \diracVector_\j, \vertex_1 + \diracVector_\j}$.
	                                  Let us show that $\translation(\vertexVector_2) = \vertexVector_2 + \diracVector_\j$.
	                                  As in the step 1), comparing the neighborhoods of $\vertexVector_2$ and $\translation(\vertexVector_1)$ gives us $\neighborhood{1}{\vertexVector_2} \cap \neighborhood{1}{\translation(\vertexVector_1)} = \set{\vertexVector_1, \vertexVector_2 + \diracVector_\j}$.
	                                  However, step 1) gives us that $\vertexVector_1$ is necessarily the image of $\vertexVector_1 - \diracVector_\j$.
	                                  Since $\translation$ is injective, it follows that $\translation(\vertexVector_2) = \vertexVector_2 + \diracVector_\j$.
	                        \end{enumerate}
	                        By induction, we conclude that for every vertex $\vertexVector_2 \in \vertices_t : \translation(\vertexVector_2) = \vertexVector_2 + \diracVector_\j$.
	                    \end{proof}
	
	                    The constraint of having all dimensions in $\gridDimensions$ being larger than $5$ allows any lossless translation to verify $\forall \vertexVector \in \vertices_t : \translation(\vertexVector) = \vertexVector + \diracVector_\j$.
	                    For smaller graphs, lossless translations can be found for which this property is not true.
	                    As an example, \figref{torusTranslation4} depicts lossless translations on a torus of dimensions $\gridDimensions = \vector{4\\4}$ that are not translations by a Dirac vector.
	                    Still, note that any translation $\translation$ such that $\forall \vertexVector \in \vertices_t : \translation(\vertexVector) = \vertexVector + \diracVector_\j$ is lossless even for smaller grid graphs.
	                    
	                    \begin{figure}
	                        \centering
	                        \begin{tikzpicture}[scale=0.6,thick]
	                          \tikzstyle{every node} = [draw, circle];
	                          \foreach \i in {0,1,2,3}{
	                            \foreach \j in {0,1,2,3}{
	                              \node(\i\j) at (\i,\j) {};
	                            }
	                          }
	                          \path[dotted]
	                          \foreach \i in {0,1,2,3}{
	                            \foreach \j/\jj in {0/1,1/2,2/3}{
	                              (\i\j) edge (\i\jj)
	                            }
	                          }
	                          \foreach \i/\ii in {0/1,1/2,2/3}{
	                            \foreach \j in {0,1,2,3}{
	                              (\i\j) edge (\ii\j)
	                            }
	                          }
	                          \foreach \i in {0,1,2,3}{
	                            (\i3) edge[bend left] (\i0)
	                          }
	                          \foreach \j in {0,1,2,3}{
	                            (0\j) edge[bend left] (3\j)
	                          }
	                          ;
	                          \path[->, red]
	                          (00) edge (01)
	                          (01) edge (11)
	                          (11) edge (10)
	                          (10) edge (00)
	                          (30) edge (31)
	                          (31) edge (21)
	                          (21) edge (20)
	                          (20) edge (30)
	                          (03) edge (02)
	                          (02) edge (12)
	                          (12) edge (13)
	                          (13) edge (03)
	                          (33) edge (32)
	                          (32) edge (22)
	                          (22) edge (23)
	                          (23) edge (33)
	                          ;
	                        \end{tikzpicture}
	                        ~~~~~
	                        \begin{tikzpicture}[scale=0.6,thick]
	                          \tikzstyle{every node} = [draw, circle];
	                          \foreach \i in {0,1,2,3}{
	                            \foreach \j in {0,1,2,3}{
	                              \node(\i\j) at (\i,\j) {};
	                            }
	                          }
	                          \path[dotted]
	                          \foreach \i in {0,1,2,3}{
	                            \foreach \j/\jj in {0/1,1/2,2/3}{
	                              (\i\j) edge (\i\jj)
	                            }
	                          }
	                          \foreach \i/\ii in {0/1,1/2,2/3}{
	                            \foreach \j in {0,1,2,3}{
	                              (\i\j) edge (\ii\j)
	                            }
	                          }
	                          \foreach \i in {0,1,2,3}{
	                            (\i3) edge[bend left] (\i0)
	                          };
	                          \path[->, red]
	                          \foreach \j in {0,1,2,3}{
	                            (0\j) edge[bend left] (3\j)
	                            (3\j) edge[bend right] (0\j)
	                          }
	                          \foreach \j in {0,1,2,3}{
	                            (1\j) edge (2\j)
	                            (2\j) edge (1\j)
	                          };
	                        \end{tikzpicture}
	                        \caption[]{Examples of lossless translations on a torus of dimensions $\gridDimensions = \vector{4\\4}$ that are not translations by a Dirac vector.}
	                        \label{torusTranslation4}
	                    \end{figure}
	                    
	                    Note that a direct consequence of \lemref{contaminationlemmaTorus} is that there are as many lossless translations as there are neighbors for a given vertex.
	                    By composing the lossless translations on the torus graph, we obtain more complex functions, that induce the following monoid:
	                    
	                    \begin{definition}[Monoid induced by $\losslessTranslations{\graph_t}$]
	                        We call \emph{monoid induced by $\losslessTranslations{\graph_t}$} the minimum monoid containing $\losslessTranslations{\graph_t}$ with the composition of fonctions as inner law.
	                    \end{definition}
	                    
	                    \begin{proposition}
	                        For torus graphs with $\forall \i \in \intInterval{1, \gridNbDimensions} : \gridDimensions\entry{\i} \geq 5$, the monoid induced by $\losslessTranslations{\graph_t}$ is exactly the set of Euclidean translations on the torus graph.
	                        \label{equivalenceTorus}
	                    \end{proposition}
	                    
	                    \begin{proof}
	                        A direct consequence of \lemref{contaminationlemmaTorus} is that lossless translations $\translation \in \losslessTranslations{\graph_t}$ on the torus graph can be obtained by choosing a Dirac vector for a dimension $\i \in \intInterval{1, \gridNbDimensions}$ and applying the contamination.
	                        Therefore, $\forall \i \in \intInterval{1, \gridNbDimensions} : \exists! \translation \in \losslessTranslations{\graph_t} : \forall \vertexVector \in \vertices_t : \translation(\vertexVector) = \vertexVector + \diracVector_\i$.
	                        We obtain that $\differenceVector$ in \defref{defEuclideanTranslationTorus} is a linear combination of vectors in $\set{\diracVector_1, \diracVector_2, \dots, \diracVector_\gridNbDimensions}$.
	                        As a consequence, any Euclidean translation on the torus graph can be written as a composition of lossless translations on the torus graph, which are elements of the monoid induced by $\losslessTranslations{\graph_t}$.
	                    \end{proof}

	                \subsubsection{Application to the grid graph}
	                \label{resultsGrid} ~
	                    
	                    Let us now proceed with grid graphs $\graph_g = \tuple{\vertices_g, \edges_g}$ yielded by a dimensions vector $\gridDimensions \in \strictlyPositiveIntegers^\gridNbDimensions$.
	                    We can adapt the definition of Euclidean translation on the torus graph in \defref{defEuclideanTranslationTorus} to grid graphs as follows:
	                    
	                    \begin{definition}[Euclidean translation on the grid graph]
	                        An \emph{Euclidean translation} $\translation_g$ on the grid graph $\graph_g = \tuple{\vertices_g, \edges_g}$ of dimensions $\gridDimensions \in \strictlyPositiveIntegers^\gridNbDimensions$ is such that:
	                        $$
	                            \exists \differenceVector \in \positiveIntegers^\gridNbDimensions : \forall \vertexVector \in \vertices_g : \translation(\vertexVector) =
	                            \left\{
	                                \begin{array}{ll}
	                                    \vertexVector + \differenceVector & \text{if } \vertexVector + \differenceVector \in \vertices_g \\
	                                    \blackHole & \text{otherwise}
	                                \end{array}
	                            \right.
	                            \;.
	                        $$
	                        \label{defEuclideanTranslationGrid}
	                    \end{definition}
	                    
	                    \begin{proposition}
	                        Let us consider the translation $\translation \in \translations{\graph_g}$ such that
	                        $$
	                            \forall \vertexVector \in \vertices_g : \translation(\vertexVector) =
	                            \left\{
	                                \begin{array}{ll}
	                                    \vertexVector + \diracVector_\i & \text{if } \vertexVector + \diracVector_\i \in \vertices_g \\
	                                    \blackHole & \text{otherwise}
	                                \end{array}
	                            \right.
	                            \;,
	                        $$
	                        for $\diracVector_\i$ the Dirac vector for dimension $\i$.
	                        We have $\loss(\translation) = \prod_{\j \in \intInterval{1, \gridNbDimensions}, \j \neq \i} \gridDimensions\entry{\j}$.
	                        \label{translationsGrid}
	                    \end{proposition}
	                    
	                    \begin{proof}
	                        By construction of the grid graph, two vertices are neighbors if their coordinates differ by $1$ or $-1$ along a single dimension.
	                        In particular, it is true for dimension $\i$.
	                        Therefore, any vertex $\vertexVector$ such that $\vertexVector + \diracVector_\i \not\in \vertices_g$ is such that $\vertexVector\entry{\i} = \gridDimensions\entry{\i}$.
	                        The product of dimensions that are different from $\i$ gives us the number of such vertices, hence the loss of $\translation$.
	                    \end{proof}
	                    
	                    \begin{remark}
	                        As for torus graphs, note that all the results in this section also apply when considering inverse Dirac vectors containing a single non-null entry $\i$ being $-1$.
	                    \end{remark}
	                    
	                    As for torus graphs, we can introduce the monoid induced by the translations on the grid graph as follows:
	                    
	                    \begin{definition}[Monoid induced by $\translations{\graph_g}$]
	                        We call \emph{monoid induced by $\translations{\graph_g}$} the minimum monoid containing $\translations{\graph_g}$ with the composition of fonctions as inner law.
	                    \end{definition}
	                    
	                    \begin{proposition}
	                        The monoid induced by $\translations{\graph_g}$ includes the set of Euclidean translations on the grid graph.
	                    \end{proposition}
	                    
	                    \begin{proof}
	                        Translations by Dirac vectors introduced in \propref{translationsGrid} exist for every dimension $i$.
	                        It follows that $\differenceVector$ in \defref{defEuclideanTranslationGrid} is a linear combination of vectors in $\set{\diracVector_1, \diracVector_2, \dots, \diracVector_\gridNbDimensions}$.
	                        As a consequence, any Euclidean translation on the grid graph can be written as a composition of translations on the grid graph, which are elements of the monoid induced by $\translations{\graph_g}$.
	                        
	                        However, contrary to the case of torus graphs in \propref{equivalenceTorus}, translations introduced in \propref{translationsGrid} are only a subset of $\translations{\graph_g}$.
	                        Therefore, Euclidean translations are included in the monoid induced by $\translations{\graph_g}$, but the converse is not true.
	                        
	                        As an counterexample, for the Dirac vector $\diracVector_1$ and a grid graph such that $\forall \i \in \intInterval{1, \gridNbDimensions} : \gridDimensions\entry{\i} \geq 3$, the translation $\translation \in \translations{\graph_g}$ such that
	                        $$
	                            \scalebox{0.92}{$
	                                \forall \vertexVector \in \vertices_g : \translation(\vertexVector) =
	                                \left\{
	                                    \begin{array}{ll}
	                                        \vertexVector + \diracVector_1 & \text{if } \left(\vertexVector + \diracVector_1 \in \vertices_g\right) \wedge \left(\vertexVector \neq \vector{1\\\dots\\1}\right) \\
	                                        \blackHole & \text{otherwise}
	                                    \end{array}
	                                \right.
	                            $}
	                        $$
	                        is not an Euclidean translation on the grid graph.
	                    \end{proof}
	                    
	                    Now, we are interested in showing that Euclidean translations by $\diracVector_\i$ (or $-\diracVector_\i$) are pseudo-minimal on the grid graph.
	                    We restrict our study to a subclass of grid graphs such that each dimension is large compared to the following ones, \ie,
	                    \begin{equation}
	                        \gridDimensions\entry{\gridNbDimensions} \geq 3 \wedge \forall \i \in \intInterval{1, \gridNbDimensions - 1} : \gridDimensions\entry{\i} \geq 2 + 2 \prod_{\j = \i + 1}^{\gridNbDimensions} \gridDimensions\entry{\j}
	                        \;.
	                        \label{assumptionLargeGrid}
	                    \end{equation}
	                    This hypothesis is necessary for the subsequent proofs.
	                    However, we conjecture the following result:
	                    
	                    \begin{conjecture}
	                        The forthcoming results apply for grid graphs such that $\forall \i \in \intInterval{1, \gridNbDimensions} : \gridDimensions\entry{\i} \geq 6$.
	                        \label{conjecture6}
	                    \end{conjecture}
	                    
	                    To ease exposition of the following results, we introduce the notion of \emph{slice of a grid graph} as follows:
	                    \begin{definition}[Grid graph slice]
	                        We call \emph{slice of a grid graph}, noted $\vertices_g^{\i,\j}$ the subset of vertices $\vertices_g$ such that they have their $\i\th$ coordinate equal to $\j$, \ie,
	                        $$
	                            \vertices_g^{\i,\j} = \condSet{\vertexVector \in \vertices_g}{ \vertexVector\entry{\i} = \j}
	                            \;.
	                        $$
	                    \end{definition}
	                    
	                    \begin{lemma}
	                        Let $\graph_g$ be a grid graph respecting assumption \eqref{assumptionLargeGrid}.
	                        If $\translation \in \translations{\graph_g}$ is a minimal translation, then:
	                        $$
	                            \exists \i : \forall \vertexVector \in \vertices_g^{1,\i} \cup \vertices_g^{1,\i+1} : \translation(\vertexVector) \neq \blackHole
	                            \;.
	                        $$
	                    \end{lemma}
	                    
	                    \begin{proof}
	                        By \propref{translationsGrid}, we have an upper bound on the loss of minimal translations when translating vertices along a single dimension.
	                        When considering dimension $1$, any minimal translation has therefore at most $\prod_{\j \in \intInterval{2, \gridNbDimensions}} \gridDimensions\entry{\j}$ vertices that have their image through $\translation$ being $\blackHole$.
	                        From assumption \eqref{assumptionLargeGrid}, we have that $\gridDimensions\entry{1} \geq 2 + 2 \prod_{\j \in \intInterval{2, \gridNbDimensions}} \gridDimensions\entry{\j}$.
	                        Since $\gridDimensions\entry{1} - 2 \prod_{\j \in \intInterval{2, \gridNbDimensions}} \gridDimensions\entry{\j} + 1 > 1$, there cannot be a strict alternance of vertices with image in $\vertices_g$, and vertices with image equal to $\blackHole$, as it would violate the upper bound on the loss.
	                        Therefore, there exist two slices $\vertices_g^{1,\i}$ and $\vertices_g^{1,\i+1}$ that contain no vertex $\vertexVector$ such that $\translation(\vertexVector) = \blackHole$.
	                    \end{proof}
	                    
	                    \begin{lemma}
	                        Let $\graph_g$ be a grid graph respecting assumption \eqref{assumptionLargeGrid}, and let $\translation \in \translations{\graph_g}$ be a minimal translation.
	                        If $\vertices_g^{1,\i}$ and $\vertices_g^{1,\i+1}$ are two slices containing no vertex of which image by $\translation$ is $\blackHole$, then:
	                        $$
	                            \translation(\vertices_g^{1,\i} \cup \vertices_g^{1,\i+1}) \not\subset \vertices_g^{1,\i} \cup \vertices_g^{1,\i+1}
	                            \;.
	                        $$
	                        \label{lemmaNotInSlices}
	                    \end{lemma}
	                    
	                    \begin{proof}
	                        Let us consider a vertex $\vertexVector_1 \in \vertices_g^{1,\i} \cup \vertices_g^{1,\i+1}$.
	                        \propref{translationsAreCyclesOrPaths} tells us that there are two cases to consider:
	                        \begin{enumerate}
	                            \item $\exists n : \translation^n(\vertexVector_1) = \blackHole$.
	                                  Since no vertex in $\vertices_g^{1,\i} \cup \vertices_g^{1,\i+1}$ has its image being $\blackHole$, the sequence $\sequence{\translation^n(\vertexVector_1)}_n$ necessarily contains a vertex $\vertexVector_2 \not\in \vertices_g^{1,\i} \cup \vertices_g^{1,\i+1}$.
	                            \item $\exists n : \translation^n(\vertexVector_1) = \vertexVector_1$.
	                                  In this case, we distinguish the following situations, illustrated on a $\vector{8\\3}$ grid graph:
	                                  \begin{enumerate}
	                                      \item Every vertex from $\vertices_g^{1,\i}$ is sent to its neighbor in $\vertices_g^{1,\i+1}$, and every vertex from $\vertices_g^{1,\i+1}$ is sent to its neighbor in $\vertices_g^{1,\i+2}$.
	                                            \begin{center}
	                                                \begin{tikzpicture}[thick]
	                                                  \draw [fill=gray!20, gray!20] (2.55, -1.5) rectangle (3.45, 2.4);
	                                                  \draw [fill=gray!20, gray!20] (3.55, -1.5) rectangle (4.45, 2.4);
	                                                  \node (00) [draw, fill=white, circle] at (0,0) {};
	                                                  \node (01) [draw, fill=white, circle] at (0,1) {};
	                                                  \node (02) [draw, fill=white, circle] at (0,2) {};
	                                                  \node (10) [draw, fill=white, circle] at (1,0) {};
	                                                  \node (11) [draw, fill=white, circle] at (1,1) {};
	                                                  \node (12) [draw, fill=white, circle] at (1,2) {};
	                                                  \node (20) [draw, fill=white, circle] at (2,0) {};
	                                                  \node (21) [draw, fill=white, circle] at (2,1) {};
	                                                  \node (22) [draw, fill=white, circle] at (2,2) {};
	                                                  \node (30) [draw, fill=white, circle] at (3,0) {};
	                                                  \node (31) [draw, fill=white, circle] at (3,1) {};
	                                                  \node (32) [draw, fill=white, circle] at (3,2) {};
	                                                  \node (40) [draw, fill=white, circle] at (4,0) {};
	                                                  \node (41) [draw, fill=white, circle] at (4,1) {};
	                                                  \node (42) [draw, fill=white, circle] at (4,2) {};
	                                                  \node (50) [draw, fill=white, circle] at (5,0) {};
	                                                  \node (51) [draw, fill=white, circle] at (5,1) {};
	                                                  \node (52) [draw, fill=white, circle] at (5,2) {};
	                                                  \node (60) [draw, fill=white, circle] at (6,0) {};
	                                                  \node (61) [draw, fill=white, circle] at (6,1) {};
	                                                  \node (62) [draw, fill=white, circle] at (6,2) {};
	                                                  \node (70) [draw, fill=white, circle] at (7,0) {};
	                                                  \node (71) [draw, fill=white, circle] at (7,1) {};
	                                                  \node (72) [draw, fill=white, circle] at (7,2) {};
	                                                  \node (slice1) at (3,-1) {$\vertices_g^{1,\i}$};
	                                                  \node (slice2) at (4,-1) {$\vertices_g^{1,\i+1}$};
	                                                  \path[dotted]
	                                                  \foreach \i in {0,1,2,3,4,5,6,7}{
	                                                    \foreach \j/\jj in {0/1,1/2}{
	                                                      (\i\j) edge (\i\jj)
	                                                    }
	                                                  }
	                                                  \foreach \i/\ii in {0/1,1/2,2/3,3/4,4/5,5/6,6/7}{
	                                                    \foreach \j in {0,1,2}{
	                                                      (\i\j) edge (\ii\j)
	                                                    }
	                                                  };
	                                                  \path[->, red]
	                                                  (30) edge (40)
	                                                  (31) edge (41)
	                                                  (32) edge (42)
	                                                  (40) edge (50)
	                                                  (41) edge (51)
	                                                  (42) edge (52);
	                                                \end{tikzpicture}
	                                            \end{center}
	                                            In this situation, there cannot exist a cycle such that $\translation^n(\vertexVector_1) = \vertexVector_1$ due to injectivity of $\translation$.
	                                            Note that this situation also applies in the case where every vertex from $\vertices_g^{1,\i}$ is sent to its neighbor in $\vertices_g^{1,\i-1}$, and every vertex from $\vertices_g^{1,\i+1}$ is sent to its neighbor in $\vertices_g^{1,\i}$.
	                                      \item Every vertex from $\vertices_g^{1,\i}$ is sent to its neighbor in $\vertices_g^{1,\i+1}$, and every vertex from $\vertices_g^{1,\i+1}$ is sent to its neighbor in $\vertices_g^{1,\i}$.
	                                            \begin{center}
	                                                \begin{tikzpicture}[thick]
	                                                  \draw [fill=gray!20, gray!20] (2.55, -1.5) rectangle (3.45, 2.4);
	                                                  \draw [fill=gray!20, gray!20] (3.55, -1.5) rectangle (4.45, 2.4);
	                                                  \node (00) [draw, fill=white, circle] at (0,0) {};
	                                                  \node (01) [draw, fill=white, circle] at (0,1) {};
	                                                  \node (02) [draw, fill=white, circle] at (0,2) {};
	                                                  \node (10) [draw, fill=white, circle] at (1,0) {};
	                                                  \node (11) [draw, fill=white, circle] at (1,1) {};
	                                                  \node (12) [draw, fill=white, circle] at (1,2) {};
	                                                  \node (20) [fill=black, circle] at (2,0) {};
	                                                  \node (21) [fill=black, circle] at (2,1) {};
	                                                  \node (22) [fill=black, circle] at (2,2) {};
	                                                  \node (30) [draw, fill=white, circle] at (3,0) {};
	                                                  \node (31) [draw, fill=white, circle] at (3,1) {};
	                                                  \node (32) [draw, fill=white, circle] at (3,2) {};
	                                                  \node (40) [draw, fill=white, circle] at (4,0) {};
	                                                  \node (41) [draw, fill=white, circle] at (4,1) {};
	                                                  \node (42) [draw, fill=white, circle] at (4,2) {};
	                                                  \node (50) [fill=black, circle] at (5,0) {};
	                                                  \node (51) [fill=black, circle] at (5,1) {};
	                                                  \node (52) [fill=black, circle] at (5,2) {};
	                                                  \node (60) [draw, fill=white, circle] at (6,0) {};
	                                                  \node (61) [draw, fill=white, circle] at (6,1) {};
	                                                  \node (62) [draw, fill=white, circle] at (6,2) {};
	                                                  \node (70) [draw, fill=white, circle] at (7,0) {};
	                                                  \node (71) [draw, fill=white, circle] at (7,1) {};
	                                                  \node (72) [draw, fill=white, circle] at (7,2) {};
	                                                  \node (slice1) at (3,-1) {$\vertices_g^{1,\i}$};
	                                                  \node (slice2) at (4,-1) {$\vertices_g^{1,\i+1}$};
	                                                  \path[dotted]
	                                                  \foreach \i in {0,1,2,3,4,5,6,7}{
	                                                    \foreach \j/\jj in {0/1,1/2}{
	                                                      (\i\j) edge (\i\jj)
	                                                    }
	                                                  }
	                                                  \foreach \i/\ii in {0/1,1/2,2/3,3/4,4/5,5/6,6/7}{
	                                                    \foreach \j in {0,1,2}{
	                                                      (\i\j) edge (\ii\j)
	                                                    }
	                                                  };
	                                                  \path[->, red]
	                                                  (30) edge (40)
	                                                  (31) edge (41)
	                                                  (32) edge (42)
	                                                  (40) edge (30)
	                                                  (41) edge (31)
	                                                  (42) edge (32);
	                                                \end{tikzpicture}
	                                            \end{center}
	                                            This causes all vertices in $\vertices_g^{1,\i-1} \cup \vertices_g^{1,\i+2}$ to be sent to $\blackHole$, leading to a loss twice higher than the upper bound for minimal translations given in \propref{translationsGrid}.
	                                            Therefore we reach a contradiction.
	                                      \item There exists a vertex $\vertexVector_2 \in \vertices_g^{1,\i}$ such that $\translation(\vertex_2) \in \vertices_g^{1,\i+1}$, and a vertex There exists a vertex $\vertexVector_3 \in \vertices_g^{1,\i+1} \cap \neighborhood{1}{\translation(\vertex_2)}$ such that $\translation(\vertex_3) \in \vertices_g^{1,\i}$.
	                                            Necessarily, $\translation^2(\vertexVector_2) = \vertexVector_3$ and $\translation^2(\vertexVector_3) = \vertexVector_2$, causing apparition of a cycle of $4$ vertices.
	                                            \begin{center}
	                                                \begin{tikzpicture}[thick]
	                                                  \draw [fill=gray!20, gray!20] (2.55, -1.5) rectangle (3.45, 2.4);
	                                                  \draw [fill=gray!20, gray!20] (3.55, -1.5) rectangle (4.45, 2.4);
	                                                  \node (00) [draw, fill=white, circle] at (0,0) {};
	                                                  \node (01) [draw, fill=white, circle] at (0,1) {};
	                                                  \node (02) [draw, fill=white, circle] at (0,2) {};
	                                                  \node (10) [draw, fill=white, circle] at (1,0) {};
	                                                  \node (11) [draw, fill=white, circle] at (1,1) {};
	                                                  \node (12) [draw, fill=white, circle] at (1,2) {};
	                                                  \node (20) [draw, fill=white, circle] at (2,0) {};
	                                                  \node (21) [draw, fill=white, circle] at (2,1) {};
	                                                  \node (22) [draw, fill=white, circle] at (2,2) {};
	                                                  \node (30) [draw, fill=white, circle] at (3,0) {};
	                                                  \node (31) [draw, fill=white, circle] at (3,1) {};
	                                                  \node (32) [draw, fill=black, circle] at (3,2) {};
	                                                  \node (40) [draw, fill=white, circle] at (4,0) {};
	                                                  \node (41) [draw, fill=white, circle] at (4,1) {};
	                                                  \node (42) [draw, fill=white, circle] at (4,2) {};
	                                                  \node (50) [draw, fill=white, circle] at (5,0) {};
	                                                  \node (51) [draw, fill=white, circle] at (5,1) {};
	                                                  \node (52) [draw, fill=white, circle] at (5,2) {};
	                                                  \node (60) [draw, fill=white, circle] at (6,0) {};
	                                                  \node (61) [draw, fill=white, circle] at (6,1) {};
	                                                  \node (62) [draw, fill=white, circle] at (6,2) {};
	                                                  \node (70) [draw, fill=white, circle] at (7,0) {};
	                                                  \node (71) [draw, fill=white, circle] at (7,1) {};
	                                                  \node (72) [draw, fill=white, circle] at (7,2) {};
	                                                  \node (slice1) at (3,-1) {$\vertices_g^{1,\i}$};
	                                                  \node (slice2) at (4,-1) {$\vertices_g^{1,\i+1}$};
	                                                  \path[dotted]
	                                                  \foreach \i in {0,1,2,3,4,5,6,7}{
	                                                    \foreach \j/\jj in {0/1,1/2}{
	                                                      (\i\j) edge (\i\jj)
	                                                    }
	                                                  }
	                                                  \foreach \i/\ii in {0/1,1/2,2/3,3/4,4/5,5/6,6/7}{
	                                                    \foreach \j in {0,1,2}{
	                                                      (\i\j) edge (\ii\j)
	                                                    }
	                                                  };
	                                                  \path[->, red]
	                                                  (30) edge (40)
	                                                  (41) edge (31);
	                                                  \path[->, red, dotted]
	                                                  (40) edge (41)
	                                                  (31) edge (30)
	                                                  (42) edge (32);
	                                                \end{tikzpicture}
	                                            \end{center}
	                                            Then, since all dimensions are larger than $3$, at least one of the vertices from this cycle has a neighbor $\vertexVector_4 \in \vertices_g^{1,\i} \cup \vertices_g^{1,\i+1}$ for which neighborhood cannot be preserved.
	                                            As a consequence, there exists a vertex in $\vertices_g^{1,\i} \cup \vertices_g^{1,\i+1}$ that has its image equal to $\blackHole$, and we reach a contradiction.
	                                            
	                                            Additionally, note that in the case where the diedges of opposite directions are not adjacent, there is necessarily at least a vertex of which image is $\blackHole$ between them.
	                                      \item Every vertex from $\vertices_g^{1,\i}$ (resp. $\vertices_g^{1,\i+1}$) is sent to a neighbor in $\vertices_g^{1,\i}$ (resp. $\vertices_g^{1,\i+1}$).
	                                            If the corresponding diedges are of opposite directions, this situation eventually leads to a \emph{turn}, in this case situation c) concludes.
	                                            If they take the same direction, then due to border effects, at least a vertex in $\vertices_g^{1,\i} \cup \vertices_g^{1,\i+1}$ has its image equal to $\blackHole$, leading to a contradiction.
	                                            \begin{center}
	                                                \begin{tikzpicture}[thick]
	                                                  \draw [fill=gray!20, gray!20] (2.55, -1.5) rectangle (3.45, 2.4);
	                                                  \draw [fill=gray!20, gray!20] (3.55, -1.5) rectangle (4.45, 2.4);
	                                                  \node (00) [draw, fill=white, circle] at (0,0) {};
	                                                  \node (01) [draw, fill=white, circle] at (0,1) {};
	                                                  \node (02) [draw, fill=white, circle] at (0,2) {};
	                                                  \node (10) [draw, fill=white, circle] at (1,0) {};
	                                                  \node (11) [draw, fill=white, circle] at (1,1) {};
	                                                  \node (12) [draw, fill=white, circle] at (1,2) {};
	                                                  \node (20) [draw, fill=white, circle] at (2,0) {};
	                                                  \node (21) [draw, fill=white, circle] at (2,1) {};
	                                                  \node (22) [draw, fill=white, circle] at (2,2) {};
	                                                  \node (30) [draw, fill=white, circle] at (3,0) {};
	                                                  \node (31) [draw, fill=white, circle] at (3,1) {};
	                                                  \node (32) [draw, fill=black, circle] at (3,2) {};
	                                                  \node (40) [draw, fill=white, circle] at (4,0) {};
	                                                  \node (41) [draw, fill=white, circle] at (4,1) {};
	                                                  \node (42) [draw, fill=white, circle] at (4,2) {};
	                                                  \node (50) [draw, fill=white, circle] at (5,0) {};
	                                                  \node (51) [draw, fill=black, circle] at (5,1) {};
	                                                  \node (52) [draw, fill=white, circle] at (5,2) {};
	                                                  \node (60) [draw, fill=white, circle] at (6,0) {};
	                                                  \node (61) [draw, fill=white, circle] at (6,1) {};
	                                                  \node (62) [draw, fill=white, circle] at (6,2) {};
	                                                  \node (70) [draw, fill=white, circle] at (7,0) {};
	                                                  \node (71) [draw, fill=white, circle] at (7,1) {};
	                                                  \node (72) [draw, fill=white, circle] at (7,2) {};
	                                                  \node (slice1) at (3,-1) {$\vertices_g^{1,\i}$};
	                                                  \node (slice2) at (4,-1) {$\vertices_g^{1,\i+1}$};
	                                                  \path[dotted]
	                                                  \foreach \i in {0,1,2,3,4,5,6,7}{
	                                                    \foreach \j/\jj in {0/1,1/2}{
	                                                      (\i\j) edge (\i\jj)
	                                                    }
	                                                  }
	                                                  \foreach \i/\ii in {0/1,1/2,2/3,3/4,4/5,5/6,6/7}{
	                                                    \foreach \j in {0,1,2}{
	                                                      (\i\j) edge (\ii\j)
	                                                    }
	                                                  };
	                                                  \path[->, red]
	                                                  (30) edge (31)
	                                                  (31) edge (32)
	                                                  (40) edge (41)
	                                                  (41) edge (42)
	                                                  (42) edge (52);
	                                                \end{tikzpicture}
	                                            \end{center}
	                                            \begin{center}
	                                                \begin{tikzpicture}[thick]
	                                                  \draw [fill=gray!20, gray!20] (2.55, -1.5) rectangle (3.45, 2.4);
	                                                  \draw [fill=gray!20, gray!20] (3.55, -1.5) rectangle (4.45, 2.4);
	                                                  \node (00) [draw, fill=white, circle] at (0,0) {};
	                                                  \node (01) [draw, fill=white, circle] at (0,1) {};
	                                                  \node (02) [draw, fill=white, circle] at (0,2) {};
	                                                  \node (10) [draw, fill=white, circle] at (1,0) {};
	                                                  \node (11) [draw, fill=white, circle] at (1,1) {};
	                                                  \node (12) [draw, fill=white, circle] at (1,2) {};
	                                                  \node (20) [draw, fill=white, circle] at (2,0) {};
	                                                  \node (21) [draw, fill=white, circle] at (2,1) {};
	                                                  \node (22) [draw, fill=white, circle] at (2,2) {};
	                                                  \node (30) [draw, fill=white, circle] at (3,0) {};
	                                                  \node (31) [draw, fill=black, circle] at (3,1) {};
	                                                  \node (32) [draw, fill=white, circle] at (3,2) {};
	                                                  \node (40) [draw, fill=white, circle] at (4,0) {};
	                                                  \node (41) [draw, fill=white, circle] at (4,1) {};
	                                                  \node (42) [draw, fill=white, circle] at (4,2) {};
	                                                  \node (50) [draw, fill=white, circle] at (5,0) {};
	                                                  \node (51) [draw, fill=white, circle] at (5,1) {};
	                                                  \node (52) [draw, fill=black, circle] at (5,2) {};
	                                                  \node (60) [draw, fill=white, circle] at (6,0) {};
	                                                  \node (61) [draw, fill=white, circle] at (6,1) {};
	                                                  \node (62) [draw, fill=white, circle] at (6,2) {};
	                                                  \node (70) [draw, fill=white, circle] at (7,0) {};
	                                                  \node (71) [draw, fill=white, circle] at (7,1) {};
	                                                  \node (72) [draw, fill=white, circle] at (7,2) {};
	                                                  \node (slice1) at (3,-1) {$\vertices_g^{1,\i}$};
	                                                  \node (slice2) at (4,-1) {$\vertices_g^{1,\i+1}$};
	                                                  \path[dotted]
	                                                  \foreach \i in {0,1,2,3,4,5,6,7}{
	                                                    \foreach \j/\jj in {0/1,1/2}{
	                                                      (\i\j) edge (\i\jj)
	                                                    }
	                                                  }
	                                                  \foreach \i/\ii in {0/1,1/2,2/3,3/4,4/5,5/6,6/7}{
	                                                    \foreach \j in {0,1,2}{
	                                                      (\i\j) edge (\ii\j)
	                                                    }
	                                                  };
	                                                  \path[->, red]
	                                                  (30) edge (31)
	                                                  (40) edge (41)
	                                                  (41) edge (42)
	                                                  (42) edge (32);
	                                                \end{tikzpicture}
	                                            \end{center}
	                                            \begin{center}
	                                                \begin{tikzpicture}[thick]
	                                                  \draw [fill=gray!20, gray!20] (2.55, -1.5) rectangle (3.45, 2.4);
	                                                  \draw [fill=gray!20, gray!20] (3.55, -1.5) rectangle (4.45, 2.4);
	                                                  \node (00) [draw, fill=white, circle] at (0,0) {};
	                                                  \node (01) [draw, fill=white, circle] at (0,1) {};
	                                                  \node (02) [draw, fill=white, circle] at (0,2) {};
	                                                  \node (10) [draw, fill=white, circle] at (1,0) {};
	                                                  \node (11) [draw, fill=white, circle] at (1,1) {};
	                                                  \node (12) [draw, fill=white, circle] at (1,2) {};
	                                                  \node (20) [draw, fill=white, circle] at (2,0) {};
	                                                  \node (21) [draw, fill=white, circle] at (2,1) {};
	                                                  \node (22) [draw, fill=white, circle] at (2,2) {};
	                                                  \node (30) [draw, fill=white, circle] at (3,0) {};
	                                                  \node (31) [draw, fill=white, circle] at (3,1) {};
	                                                  \node (32) [draw, fill=black, circle] at (3,2) {};
	                                                  \node (40) [draw, fill=white, circle] at (4,0) {};
	                                                  \node (41) [draw, fill=white, circle] at (4,1) {};
	                                                  \node (42) [draw, fill=black, circle] at (4,2) {};
	                                                  \node (50) [draw, fill=white, circle] at (5,0) {};
	                                                  \node (51) [draw, fill=white, circle] at (5,1) {};
	                                                  \node (52) [draw, fill=white, circle] at (5,2) {};
	                                                  \node (60) [draw, fill=white, circle] at (6,0) {};
	                                                  \node (61) [draw, fill=white, circle] at (6,1) {};
	                                                  \node (62) [draw, fill=white, circle] at (6,2) {};
	                                                  \node (70) [draw, fill=white, circle] at (7,0) {};
	                                                  \node (71) [draw, fill=white, circle] at (7,1) {};
	                                                  \node (72) [draw, fill=white, circle] at (7,2) {};
	                                                  \node (slice1) at (3,-1) {$\vertices_g^{1,\i}$};
	                                                  \node (slice2) at (4,-1) {$\vertices_g^{1,\i+1}$};
	                                                  \path[dotted]
	                                                  \foreach \i in {0,1,2,3,4,5,6,7}{
	                                                    \foreach \j/\jj in {0/1,1/2}{
	                                                      (\i\j) edge (\i\jj)
	                                                    }
	                                                  }
	                                                  \foreach \i/\ii in {0/1,1/2,2/3,3/4,4/5,5/6,6/7}{
	                                                    \foreach \j in {0,1,2}{
	                                                      (\i\j) edge (\ii\j)
	                                                    }
	                                                  };
	                                                  \path[->, red]
	                                                  (30) edge (31)
	                                                  (31) edge (32)
	                                                  (40) edge (41)
	                                                  (41) edge (42);
	                                                \end{tikzpicture}
	                                            \end{center}
	                                  \end{enumerate}
	                                  Note that all these situations lead to a contradiction.
	                                  Therefore, a minimal translation $\translation$ cannot lead to the creation of a cycle such that $\translation^n(\vertexVector_1) = \vertexVector_1$.
	                                  As a consequence, only case 1) applies, and $\translation(\vertices_g^{1,\i} \cup \vertices_g^{1,\i+1}) \not\subset \vertices_g^{1,\i} \cup \vertices_g^{1,\i+1}$.
	                        \end{enumerate}
	                    \end{proof}
	                    
	                    \begin{corollary}
	                        Let $\graph_g$ be a grid graph respecting assumption \eqref{assumptionLargeGrid}, and let $\translation \in \translations{\graph_g}$ be a minimal translation.
	                        Let $\vertices_g^{1,\i}$ and $\vertices_g^{1,\i+1}$ be two slices containing no vertex of which image through $\translation$ is $\blackHole$.
	                        Every vertex from $\vertices_g^{1,\i}$ is sent to its neighbor in $\vertices_g^{1,\i+1}$, and every vertex from $\vertices_g^{1,\i+1}$ is sent to its neighbor in $\vertices_g^{1,\i+2}$.
	                        \label{everyVertexToRight}
	                    \end{corollary}
	                    
	                    \begin{proof}
	                        This is a direct consequence of the proof of \lemref{lemmaNotInSlices}.
	                        Any other case corresponds to the situations described by cases 2b), 2c) and 2d) of the proof of \lemref{lemmaNotInSlices}, leading to existence of vertices of which image through $\translation$ is $\blackHole$ in $\vertices_g^{1,\i} \cup \vertices_g^{1,\i+1}$.
	                    \end{proof}
	                    
	                    \begin{lemma}
	                        Let $\graph_g$ be a grid graph respecting assumption \eqref{assumptionLargeGrid}, and let $\translation \in \translations{\graph_g}$ be a minimal translation.
	                        Let $\vertices_g^{1,\i}$ and $\vertices_g^{1,\i+1}$ be two slices containing no vertex of which image by $\translation$ is $\blackHole$:
	                        $$
	                            \forall \j \in \intInterval{1, \i+1} : \forall \vertexVector \in \vertices_g^{1,\j} : \translation(\vertexVector) \neq \blackHole
	                            \;.
	                        $$
	                        \label{lemmaBeforeSliceIsOK}
	                    \end{lemma}
	                    
	                    \begin{proof}
	                        From the proof of \lemref{lemmaNotInSlices}, there cannot exist any cycle including vertices in $\vertices_g^{1,\i} \cup \vertices_g^{1,\i+1}$.
	                        As a consequence, for every vertex $\vertexVector_1 \in \vertices_g^{1,\i}$, the sequence $\sequence{\translation^n(\vertexVector_1)}_n$ eventually leads to $\blackHole$.
	                        Since the cardinal of $\vertices_g^{1,\i}$ is $\prod_{\j \in \intInterval{2, \gridNbDimensions}} \gridDimensions\entry{\j}$, there cannot exist a vertex $\vertexVector_2$ in slices $\vertices_g^{1,\j}$ ($\j < \i$) such that $\translation(\vertexVector_2) = \blackHole$, since $\translation$ would not be minimal.
	                    \end{proof}
	                    
	                    \begin{proposition}
	                        Let $\translation_1 \in \translations{\graph_g}$ be the Euclidean translation by $\diracVector_1$ as introduced in \propref{translationsGrid} on a grid graph $\graph_g$ respecting assumption \eqref{assumptionLargeGrid}.
	                        We have that $\translation_1$ is minimal.
	                        \label{e1isMinimal}
	                    \end{proposition}
	                    
	                    \begin{proof}
	                        Let $\translation_2 \in \translations{\graph_g}$ be a minimal translation on $\graph_g$.
	                        Let $\vertices_g^{1,\i}$ and $\vertices_g^{1,\i+1}$ be two slices containing no vertex of which image through $\translation_2$ is $\blackHole$.
	                        \lemref{lemmaBeforeSliceIsOK} tells us that no vertex $\vertexVector$ in slices $\vertices_g^{1,\j}$ ($\j \leq \i+1$) has its image equal to $\blackHole$.
	                        Additionally \lemref{lemmaNotInSlices} and \corref{everyVertexToRight} indicate that for these vertices, $\translation_2(\vertexVector) = \vertexVector + \diracVector_1$.
	                        
	                        Now, let us consider the minimum $\k > \i+1$ such that $\exists \vertexVector_1 \in \vertices_g^{1,\k} : \translation_2(\vertexVector_1) = \blackHole$.
	                        \corref{everyVertexToRight} tells us that every vertex from $\vertices_g^{1,\k-1}$ has its image through $\translation_2$ in $\vertices_g^{1,\k}$.
	                        Now, let us distinguish two cases:
	                        \begin{enumerate}
	                            \item If $\k = \gridDimensions\entry{1}$, then $\translation_2 = \translation_1$, for which the loss is equal to the upper bound on losses for minimal translations;
	                            \item If $\k < \gridDimensions\entry{1}$, then we proceed by contradiction.
	                                  By \corref{everyVertexToRight}, we have that $\translation_2(\vertexVector_1 - \diracVector_1) = \vertexVector_1$.
	                                  Since $\neighborhood{1}{\vertexVector_1 - \diracVector_1} \cap \neighborhood{1}{\vertexVector_1 + \diracVector_1} = \set{\vertexVector_1}$, and since $\translation_2(\vertexVector_1) = \blackHole$, we obtain that $\vertexVector_1 + \diracVector_1$ cannot be the image of any vertex.
	                                  As a consequence, we have a sequence $\sequence{\translation_2^n(\vertexVector_1 + \diracVector_1)}_n$ that ends with $\blackHole$.
	                                  Therefore, the loss of $\translation_2$ is at least $1 + \prod_{\j \in \intInterval{2, \gridNbDimensions}} \gridDimensions\entry{\j}$, and $\translation_2$ is not minimal.
	                        \end{enumerate}
	                    \end{proof}
	                    
	                    \begin{corollary}
	                        Let $\graph_g$ be a grid graph respecting assumption \eqref{assumptionLargeGrid}.
	                        Euclidean translations by Dirac vectors $\diracVector_\i$ ($\i \in \intInterval{1, \gridNbDimensions}$) are pseudo-minimal.
	                    \end{corollary}
	                    
	                    \begin{proof}
	                        Let us denote by $\translation_\i$ the translation of every vertex by $\diracVector_i$ as introduced in \propref{translationsGrid}.
	                        \propref{e1isMinimal} shows that $\translation_1$ is minimal, hence pseudo-minimal.
	                        
	                        Now, let us consider a translation $\translation \in \translations{\graph_g}$ such that $\forall \vertexVector \in \vertices_g : \translation(\vertexVector) \neq \vertexVector + \diracVector_1$.
	                        For such translations, we can perform the same reasoning as above, leading to $\translation_2$ being pseudo-minimal.
	                        However, in the case where such a vertex exists, we can have the situation where $\translation_2 \prec \translation$.
	                        The following illustration depicts such a possible $\translation$:
	                        
	                        \begin{center}
	                            \begin{tikzpicture}[scale=0.6, thick]
	                              \node (00) [draw, circle] at (0,0) {};
	                              \node (01) [draw, circle] at (0,1) {};
	                              \node (02) [fill=black, circle] at (0,2) {};
	                              \node (10) [fill=black, circle] at (1,0) {};
	                              \node (11) [fill=black, circle] at (1,1) {};
	                              \node (12) [draw, circle] at (1,2) {};
	                              \node (20) [draw, circle] at (2,0) {};
	                              \node (21) [draw, circle] at (2,1) {};
	                              \node (22) [draw, circle] at (2,2) {};
	                              \node (30) [draw, circle] at (3,0) {};
	                              \node (31) [draw, circle] at (3,1) {};
	                              \node (32) [draw, circle] at (3,2) {};
	                              \node (40) [draw, circle] at (4,0) {};
	                              \node (41) [draw, circle] at (4,1) {};
	                              \node (42) [draw, circle] at (4,2) {};
	                              \node (50) [draw, circle] at (5,0) {};
	                              \node (51) [draw, circle] at (5,1) {};
	                              \node (52) [draw, circle] at (5,2) {};
	                              \node (60) [draw, circle] at (6,0) {};
	                              \node (61) [draw, circle] at (6,1) {};
	                              \node (62) [draw, circle] at (6,2) {};
	                              \node (70) [fill=black, circle] at (7,0) {};
	                              \node (71) [fill=black, circle] at (7,1) {};
	                              \node (72) [fill=black, circle] at (7,2) {};
	                              \path[dotted]
	                              \foreach \i in {0,1,2,3,4,5,6,7}{
	                                \foreach \j/\jj in {0/1,1/2}{
	                                  (\i\j) edge (\i\jj)
	                                }
	                              }
	                              \foreach \i/\ii in {0/1,1/2,2/3,3/4,4/5,5/6,6/7}{
	                                \foreach \j in {0,1,2}{
	                                  (\i\j) edge (\ii\j)
	                                }
	                              };
	                              \path[->, red]
	                              (00) edge (01)
	                              (01) edge (02)
	                              (20) edge (30)
	                              (30) edge (40)
	                              (40) edge (50)
	                              (50) edge (60)
	                              (60) edge (70)
	                              (21) edge (31)
	                              (31) edge (41)
	                              (41) edge (51)
	                              (51) edge (61)
	                              (61) edge (71)
	                              (12) edge (22)
	                              (22) edge (32)
	                              (32) edge (42)
	                              (42) edge (52)
	                              (52) edge (62)
	                              (62) edge (72);
	                            \end{tikzpicture}
	                        \end{center}
	                        
	                        However, in this situation, $\translation$ is not pseudo-minimal since $\translation \prec \translation_1$.
	                        It follows that $\translation_2$ is pseudo-minimal.
	                        The same reasoning can be made for all higher dimensions.
	                    \end{proof}
	                    
	                    Finally, recall from \conjref{conjecture6} that we believe all the results in this section apply for grid graphs such that $\forall \i \in \intInterval{1, \gridNbDimensions} : \gridDimensions\entry{\i} \geq 6$.
	                    Interestingly, for smaller dimensions, counterexamples as depicted in \figref{counterexamples345} can be found.
	                    The depicted translations are minimal, while not being Euclidean translations by $\diracVector_i$ as introduced in \propref{translationsGrid}.
	
	                    \begin{figure}
	                        \centering
	                        \begin{tikzpicture}[scale=0.6, thick]
	                          \node (label) [] at (1,0) {(a)};
	                          \node (00) [draw, circle] at (0,2) {};
	                          \node (01) [draw, circle] at (0,3) {};
	                          \node (02) [draw, circle] at (0,4) {};
	                          \node (10) [draw, circle] at (1,2) {};
	                          \node (11) [fill=black, circle] at (1,3) {};
	                          \node (12) [draw, circle] at (1,4) {};
	                          \node (20) [draw, circle] at (2,2) {};
	                          \node (21) [draw, circle] at (2,3) {};
	                          \node (22) [draw, circle] at (2,4) {};
	                          \path[dotted]
	                          \foreach \i in {0,1,2}{
	                            \foreach \j/\jj in {0/1,1/2}{
	                              (\i\j) edge (\i\jj)
	                            }
	                          }
	                          \foreach \i/\ii in {0/1,1/2}{
	                            \foreach \j in {0,1,2}{
	                              (\i\j) edge (\ii\j)
	                            }
	                          };
	                          \path[->, red]
	                          (20) edge (10)
	                          (10) edge (00)
	                          (00) edge (01)
	                          (01) edge (02)
	                          (02) edge (12)
	                          (12) edge (22)
	                          (22) edge (21)
	                          (21) edge (20);
	                        \end{tikzpicture}
	                        ~~~
	                        \begin{tikzpicture}[scale=0.6, thick]
	                          \node (label) [] at (1.5,0) {(b)};
	                          \node (00) [fill=black, circle] at (0,1.5) {};
	                          \node (01) [draw, circle] at (0,2.5) {};
	                          \node (02) [draw, circle] at (0,3.5) {};
	                          \node (03) [draw, circle] at (0,4.5) {};
	                          \node (10) [draw, circle] at (1,1.5) {};
	                          \node (11) [draw, circle] at (1,2.5) {};
	                          \node (12) [fill=black, circle] at (1,3.5) {};
	                          \node (13) [draw, circle] at (1,4.5) {};
	                          \node (20) [draw, circle] at (2,1.5) {};
	                          \node (21) [draw, circle] at (2,2.5) {};
	                          \node (22) [fill=black, circle] at (2,3.5) {};
	                          \node (23) [draw, circle] at (2,4.5) {};
	                          \node (30) [draw, circle] at (3,1.5) {};
	                          \node (31) [draw, circle] at (3,2.5) {};
	                          \node (32) [draw, circle] at (3,3.5) {};
	                          \node (33) [draw, circle] at (3,4.5) {};
	                          \path[dotted]
	                          \foreach \i in {0,1,2,3}{
	                            \foreach \j/\jj in {0/1,1/2,2/3}{
	                              (\i\j) edge (\i\jj)
	                            }
	                          }
	                          \foreach \i/\ii in {0/1,1/2,2/3}{
	                            \foreach \j in {0,1,2,3}{
	                              (\i\j) edge (\ii\j)
	                            }
	                          };
	                          \path[->, red]
	                          (01) edge (02)
	                          (02) edge (03)
	                          (03) edge (13)
	                          (13) edge (23)
	                          (23) edge (33)
	                          (33) edge (32)
	                          (32) edge (31)
	                          (31) edge (21)
	                          (21) edge (11)
	                          (11) edge (01)
	                          (10) edge (00)
	                          (20) edge (10)
	                          (30) edge (20);
	                        \end{tikzpicture}
	                        ~~~
	                        \begin{tikzpicture}[scale=0.6, thick]
	                          \node (label) [] at (2,0) {(c)};
	                          \node (00) [fill=black, circle] at (0,1) {};
	                          \node (01) [fill=black, circle] at (0,2) {};
	                          \node (02) [draw, circle] at (0,3) {};
	                          \node (03) [draw, circle] at (0,4) {};
	                          \node (04) [draw, circle] at (0,5) {};
	                          \node (10) [draw, circle] at (1,1) {};
	                          \node (11) [draw, circle] at (1,2) {};
	                          \node (12) [draw, circle] at (1,3) {};
	                          \node (13) [fill=black, circle] at (1,4) {};
	                          \node (14) [draw, circle] at (1,5) {};
	                          \node (20) [draw, circle] at (2,1) {};
	                          \node (21) [draw, circle] at (2,2) {};
	                          \node (22) [draw, circle] at (2,3) {};
	                          \node (23) [fill=black, circle] at (2,4) {};
	                          \node (24) [draw, circle] at (2,5) {};
	                          \node (30) [draw, circle] at (3,1) {};
	                          \node (31) [draw, circle] at (3,2) {};
	                          \node (32) [draw, circle] at (3,3) {};
	                          \node (33) [fill=black, circle] at (3,4) {};
	                          \node (34) [draw, circle] at (3,5) {};
	                          \node (40) [draw, circle] at (4,1) {};
	                          \node (41) [draw, circle] at (4,2) {};
	                          \node (42) [draw, circle] at (4,3) {};
	                          \node (43) [draw, circle] at (4,4) {};
	                          \node (44) [draw, circle] at (4,5) {};
	                          \path[dotted]
	                          \foreach \i in {0,1,2,3,4}{
	                            \foreach \j/\jj in {0/1,1/2,2/3,3/4}{
	                              (\i\j) edge (\i\jj)
	                            }
	                          }
	                          \foreach \i/\ii in {0/1,1/2,2/3,3/4}{
	                            \foreach \j in {0,1,2,3,4}{
	                              (\i\j) edge (\ii\j)
	                            }
	                          };
	                          \path[->, red]
	                          (02) edge (03)
	                          (03) edge (04)
	                          (04) edge (14)
	                          (14) edge (24)
	                          (24) edge (34)
	                          (34) edge (44)
	                          (44) edge (43)
	                          (43) edge (42)
	                          (42) edge (32)
	                          (32) edge (22)
	                          (22) edge (12)
	                          (12) edge (02)
	                          (11) edge (01)
	                          (21) edge (11)
	                          (31) edge (21)
	                          (41) edge (31)
	                          (10) edge (00)
	                          (20) edge (10)
	                          (30) edge (20)
	                          (40) edge (30);
	                        \end{tikzpicture}
	                        \caption[]
	                        {
	                            Counterexamples for grid graphs of dimensions $\vector{3\\3}$ (a), $\vector{4\\4}$ (b) and $\vector{5\\5}$ (c).
	                            For such graphs, the translations that are depicted are minimal, while not being Euclidean translations by $\diracVector_i$ as introduced in \propref{translationsGrid}.
	                        }
	                        \label{counterexamples345}
	                    \end{figure}

                    These results on the torus and on the grid graph show that in the classical case of signals evolving on regular topologies, there is a correspondence between Euclidean translations and pseudo-minimal translations on graphs.
                    Such loss-minimizing translations however do not require any vector space for their definition, and can therefore be identified on any particular topology.
                    
                    Still, the above translations suffer from two main drawbacks.
                    First, their definitions are too strict, and only allow border effects (\ie, loss) as side effects when translating signals.
                    Second, we have shown that their identification is an NP-complete problem.
                    In the following section, we introduce relaxed versions of these operators to tackle the first limitation, and propose a greedy approach to provide a solution to the complexity issue.

            \section{Relaxation of the proposed translations}
            \label{relaxation}            
             
                In this section, we propose a relaxation of translations that better apply to arbitrary graphs. 
                In order to tackle the complexity issue of identifying translations, we propose a greedy approach that find reasonable pseudo-translations, and also note that complexity can be highly reduced when focusing on localized signals on graphs.
            
                \subsection{Relaxation of the translations}
                \label{approxTranslationsScore}            
                    
                    Translations as introduced in \secref{definitions} are defined in a very strict way.
                    Small irregularities in the graph necessarily imply the loss of signal entries due to the translation.
                    One may want to find a balance between loss of the signal, EC violations (\ie, the possibility to send signal entries to non-neighboring vertices) and SNP violations (\ie, the possibility to deform the signal by losing some existing neighborhood or creating new ones).
                    To this end, we introduce approximate translations $\approximate{\translation}$ on a graph $\graph = \tuple{\vertices, \edges}$ as follows:
		            $$
		                \approximate{\translation} : \vertices_1 \cup \{\blackHole\} \to \vertices_2 \cup \{\blackHole\}
		            $$
                    where $\vertices_1 \subseteq \vertices$ and $\vertices_2 \subseteq \vertices$, and with $\approximate{\translation}(\blackHole) = \blackHole$.

                    Note that in the case when $\vertices_1 = \vertices_2 = \vertices$, approximate translations are simply transformations as introduced in \secref{definitions}.
                    Use of two distinct sets $\vertices_1$ and $\vertices_2$ here allows simplified exposition of further results.
                    
                    In order to encourage EC and SNP properties of these operators, we are interested in approximate translations $\approximate{\translation}$ that minimize the following score $\score(\approximate{\translation})$:
		            \begin{equation}
		                \begin{array}{lll}
		                    \score(\approximate{\translation}) & = & \alpha~\frac{1}{\cardinal{\vertices_1}} \loss(\approximate{\translation}) \\
		                                                       & + & \beta~\frac{1}{\cardinal{\vertices_{1\not\blackHole}}} \overline{EC}(\approximate{\translation}) \\
		                                                       & + & \gamma~\frac{2}{\cardinal{\vertices_{1\not\blackHole}} \left(\cardinal{\vertices_{1\not\blackHole}} - 1\right)} def(\approximate{\translation})
		                \end{array}
		                \;,
		                \label{score}
		            \end{equation}
		            with:
		            $$
		                \overline{EC}(\approximate{\translation}) = \cardinal{ \vertex \in \vertices_{1\not\blackHole} : \edge{\vertex}{\approximate{\translation}(\vertex)} \not\in \edges}
		                \;,
		            $$
		            and:
		            $$
		                def(\approximate{\translation}) = \sum\limits_{\edge{\vertex_1}{\vertex_2} \in \binom{\vertices_{1\not\blackHole}}{2}} \absoluteValue{\geodesic(\vertex_1, \vertex_2) - \geodesic\left(\approximate{\translation}(\vertex_1), \approximate{\translation}(\vertex_2)\right)}
		                \;,
		            $$
		            where $\vertices_{1\not\blackHole} = \set{\vertex \in \vertices_1 : \approximate{\translation}(\vertex) \neq \blackHole}$, and $\geodesic$ is the geodesic distance on the graph.
		            In \eqref{score}, $\alpha$, $\beta$ and $\gamma$ are parameters controlling the respective importance of 1) Percentage of function loss; 2) Percentage of EC violations; and 3) average signal deformation.
		            The various normalization terms associated with these values make these quantities independent from the choice of $\vertices_1$.
                    Experimentally, we suggest that $\alpha$ should be set to the highest value in order to minimize the approximate translation loss.
                    Then, $\gamma$ should have the second larger value so that the signal is as little deformed as possible.
                    A low value of $\beta$ should be sufficient to encourage going through a non-empty set of existing edges in the graph, thus conserving the overall orientation of the signal that is translated.
                    
                    Note that in \eqref{score}, we choose the signal deformation --- controlled by the $\gamma$ parameter --- to be measured by the average change in pairwise distances between non-lost vertices, instead of counting the number of created or removed neighborhoods as follows:
		            \begin{equation}
		                \begin{array}{ll}
		                    \overline{SNP}(\approximate{\translation}) = & | \{ \edge{\vertex_1}{\vertex_2} \in \binom{\vertices_1}{2} : \\
		                       & \bigg(\Big(\big(\edge{\vertex_1}{\vertex_2} \in \edges\big) \wedge \big(\edge{\approximate{\translation}(\vertex_1)}{\approximate{\translation}(\vertex_2)} \not\in \edges\big)\Big) \vee \\
		                       & \Big(\big(\edge{\vertex_1}{\vertex_2} \not\in \edges\big) \wedge \big(\edge{\approximate{\translation}(\vertex_1)}{\approximate{\translation}(\vertex_2)} \in \edges\big)\Big)\bigg) \\
		                       & \wedge~(\approximate{\translation}(\vertex_1) \neq \blackHole) \wedge (\approximate{\translation}(\vertex_2) \neq \blackHole) \}|
		                \end{array}
		                \;.
		                \label{snpViolations}
		            \end{equation}
                    This choice encourages $\approximate{\translation}$ to be an isometry, as \eqref{snpViolations} does not consider relative distances between disconnected components of a sparse signal to translate.
                    
                    When composing multiple approximate translations $\approximate{\translation} = \approximate{\translation}_1 \circ \approximate{\translation}_2 \circ \dots \circ \approximate{\translation}_T$ on the graph, we define the total score of the composed translation as follows:
                    \begin{equation}
                        \score(\approximate{\translation}) = \sum\limits_{t = 1}^T \score(\approximate{\translation}_i)
                        \;.
                        \label{scoreComposition}
                    \end{equation}
                    Due to the positivity of $\score$, the score of consecutive approximate translations of a same signal on the graph can only increase as they are applied to the signal.
                    This quantity thus measures the total deformation of a signal as it is translated, as a sum of individual deformations.

                    Note that this score is an overestimation of the actual total deformation.
                    As an example, consider a lossless approximate translation $\approximate{\translation}_1 \in \losslessTransformations{\graph}$ on the whole set of vertices, defined as a random permutation of vertices, and the inverse approximate translation $\approximate{\translation}_2 = \approximate{\translation}_1^{-1}$.
                    The composed approximate translation $\approximate{\translation}_2 \circ \approximate{\translation}_1$ has a positive score, although it is identity.
                    
                    Finally, we highlight that a composition $\approximate{\translation}$ of approximate translations can be evaluated in terms of signal deformation using a pair:
                    \begin{equation}
                        \left(\frac{1}{\cardinal{\vertices_1}}\loss(\approximate{\translation}), \frac{2}{\cardinal{\vertices_{1\not\blackHole}} \left(\cardinal{\vertices_{1\not\blackHole}} - 1\right)} \overline{SNP}(\approximate{\translation})\right)
                        \;.
    		                \label{evaluationPair}
    		           \end{equation}
                    Here, the second term gives the ratio of violations of the SNP property within the non-lost part of the translated signal.
                    Pareto optima with respect to these two quantities are the most interesting compositions in terms of a trade-off between loss and signal deformation.

                \subsection{Finding the best composition of approximate translations}
                \label{bestComposition}
                    
                    While being imprecise, monotonicity of the score function $\score$ with respect to the composition operator allows one to discriminate between multiple compositions of translations.
                    In more details, let us choose $\vertex_{src} \in \vertices_1$ and $\vertex_{tgt} \in \vertices$.
                    We are here interested in identifying the best composition of approximate translations $\approximate{\translation}_{\vertex_{src} \rightarrow \vertex_{tgt}} = \approximate{\translation}_1 \circ \approximate{\translation}_2 \circ \dots \circ \approximate{\translation}_T$ --- in the sense that $\score(\approximate{\translation}_{\vertex_{src} \rightarrow \vertex_{tgt}})$ is minimized --- such that $\approximate{\translation}_{\vertex_{src} \rightarrow \vertex_{tgt}}(\vertex_{src}) = \vertex_{tgt}$.
                    
                    In order to do so, we can exploit the monotonicity of $\score$ with respect to the composition operator to find the best composition of approximate translations by \emph{propagating} $\vertex_{src}$ using a Dijkstra-like algorithm \cite{dijkstra1959note}, as shown in \algref{algoBestComposition}.
                    This algorithm can be understood as a simple traversal in an abstract graph, in which each vertex denotes a configuration --- here, a assignment of elements in $\vertices \cup \set{\blackHole}$ to each vertex in $\vertices_1$ --- and edges are weighted by the deformation to go from a configuration to another, given by the score function in \eqref{score}.
                    A shortest path on this graph from an initial configuration $(\vertex_1 \mapsto \vertex_1, \dots, \vertex_\graphOrder \mapsto \vertex_\graphOrder)$ to a configuration such that $\vertex_{src} \mapsto \vertex_{tgt}$ gives the sequence of approximate translations that minimizes \eqref{scoreComposition}.
                    
		            \begin{algorithm}
		                \ForEach{$\vertex \in \vertices$}
		                {
		                    \State{$visited[\vertex] := \KwFalse$;}
		                    \State{$predecessors[\vertex] := \KwNull$;}
		                }
		                \State{$Q := \set{\tuple{0, \vertex_{src}, \KwNull, \KwNull}}$;}
                        \While{$Q \neq \emptyset$}
		                {
    		                    \State{$total\_s_{\vertex_1}, \vertex_1, pred, \approximate{\translation}_{pred \rightarrow \vertex_1} := \underline{extract\_min}(Q)$;}
		                    \If{$\KwNot~visited[\vertex_1]$}
		                    {
		                        \State{$visited[\vertex_1] := \KwTrue$;}
		                        \State{$predecessors[\vertex_1] := \tuple{pred, \approximate{\translation}_{pred \rightarrow \vertex_1}}$;}
			                    \If{$\vertex_1 = \vertex_{tgt}$}
			                    {
			                        \State{$\KwBreak$;}
			                    }
			                    \If{$\approximate{\translation}_{pred \rightarrow \vertex_1} \neq \KwNull$}
	    		                    {
	    		                        \State{$\vertices_1 := \condSet{\vertex \in \vertices}{\approximate{\translation}_{pred \rightarrow \vertex_1}(\vertex) \neq \blackHole}$;}
	    		                    }
			                    \State{$\vertices_2 := \underline{choose\_subset}(\vertices)$;}
			                    \ForEach{$\vertex_2 \in \vertices_2 \setminus \set{\vertex_1}$}
			                    {
			                        \If{$\KwNot~visited[\vertex_2]$}
		                            {
				                        \State{$\approximate{\translation}_{\vertex_1 \rightarrow \vertex_2}, s_{\approximate{\translation}_{\vertex_1 \rightarrow \vertex_2}} := \underline{minimize\_\score}(\vertex_1, \vertex_2, \graph, \vertices_1, \vertices_2)$;}
				                        \State{$total\_s_{\vertex_2} := total\_s_{\vertex_1} + s_{\approximate{\translation}_{\vertex_1 \rightarrow \vertex_2}}$;}
				                        \State{$Q := Q \cup \set{\tuple{total\_s_{\vertex_2}, \vertex_2, \vertex_1, \approximate{\translation}_{\vertex_1 \rightarrow \vertex_2}}}$;}
				                    }
			                    }
			                }
		                }
		                \Return{$predecessors$;}
	                    \caption{$\vertex_{src}\_to\_\vertex_{tgt}~(\graph = \tuple{\vertices, \edges}, \vertices_1, \vertex_{src}, \vertex_{tgt})$}
	                    \label{algoBestComposition}
	                \end{algorithm}
	                
	                In \algref{algoBestComposition}, function $\underline{extract\_min}$ returns the element from the priority queue with minimum key, here set as the current total score.
	                \algref{algoBestComposition} being a simple adaptation of a Dijkstra algorithm in which we want to minimize the total deformation \eqref{scoreComposition}, its worst-case complexity is $\complexity{\cardinal{\vertices}^2 + \cardinal{\vertices} \log \cardinal{\vertices}}$ when using a Fibonacci heap \cite{fredman1987fibonacci}.
	                Here, the $\cardinal{\vertices}^2$ term comes from the choice of $\vertices_2$ through $\underline{choose\_subset}$, which in the worst case can contain all vertices of $\vertices$.
	                The major complexity thus resides in the $\underline{minimize\_\score}$ function, that identifies the approximate translation $\approximate{\translation}_{\vertex_1 \rightarrow \vertex_2}$ minimizing $\score$ in \eqref{score}, and such that $\approximate{\translation}_{\vertex_1 \rightarrow \vertex_2}(\vertex_1) = \vertex_2$.
                    
                    Approximate translations being by definition transformations on the graph, the number of such objects is combinatorial.
                    In the general case, it is therefore not tractable to exhaustively iterate over all possible approximate translations, and one needs to approximate $\underline{minimize\_\score}$ with heuristics.
                    A greedy approach to do so is presented in \secref{greedyApproach}.

                \subsection{The case of localized signals}
                \label{localizedSignals}            
                    
                    As stated before, exhaustive enumeration of all possible approximate translations in $\underline{minimize\_\score}$ (in \algref{algoBestComposition}) is in general not tractable.
                    However, it is worth mentioning that the number of possible approximate translations heavily depends on choices of $\vertices_1$ and $\vertices_2$.
                                        
                    In particular, in the case of localized signals for which most entries are null, there is no interest in finding a translation operator for every vertex in the graph.
                    Let $\signalVector \in \reals^\graphOrder$ be a signal on a graph $\graph = \tuple{\vertices, \edges}$, such that $\norm{\signalVector}{0} \ll \graphOrder$.
                    We denote by $\signalVector\entry{\vertex}$ the value of the signal associated with vertex $\vertex \in \vertices$.
                    Let us assume that $\signalVector$ is localized on the graph, \ie, that the induced subgraph of vertices with non-zero signal entries is connected.

                    We can adapt \algref{algoBestComposition} as follows:
                    \begin{itemize}
                        \item Initially: $\vertices_1 := \condSet{\vertex \in \vertices}{\signalVector\entry{\vertex} \neq 0}$;
                        \item L. 20: $\vertices_2 := \vertices_1 \sqcup \condSet{\vertex \in \vertices}{\neighborhood{1}{\vertex} \cap \vertices_1 \neq \emptyset}$.
                    \end{itemize}
                    In that case, we are only interested in finding approximate translations that operate on the subset of vertices $\vertices_1$ on which the signal is defined.
                    At each step of the main loop in \algref{algoBestComposition}, $\vertices_2$ is the set of vertices at most $1$-hop away from a vertex in $\vertices_1$.
                    Note that this choice has some impact on the approximate translations that can be found, as it reduces the set of such functions to those that translate the signal to a close location.
                    Including $H$-hop neighbors ($H > 1$) in the definition of $\vertices_2$ enables finding approximate translations with more EC violations but possibly less SNP violations.
                    However, the number of such functions exponentially increases.
                    
                    Using the adaptations above, it may happen that exhaustive computation of approximate translations becomes tractable.
                    This may arise in the case when the graph is regular enough, with a bounded degree for each vertex.
                    Interestingly, examples of such graphs include random geometric graphs, meshes or -- more generally -- graphs built from a uniform sampling of a manifold.

                \subsection{A greedy approach to identify approximate translations}
                \label{greedyApproach}     
                
                    A possible approach to tackle the complexity of identifying interesting pseudo-translations in $\underline{minimize\_\score}$ (in \algref{algoBestComposition}) is to use a greedy algorithm.
                    The idea is to gradually increase the subset of vertices from $\vertices_1$ that are given an image through $\approximate{\translation}_{\vertex_1 \rightarrow \vertex_2}$.
                    For a chosen such subset, one can find the assignments of these vertices such that $\score(\approximate{\translation}_{\vertex_1 \rightarrow \vertex_2})$ --- restricted to the vertices that are given an image --- is minimized.
                    \algref{algoGreedy} gives a possible implementation for this heuristics.
                    
		            \begin{algorithm}
		                \ForEach{$\vertex \in \vertices_1$}
                        {
		                    \State{$\approximate{\translation}_{\vertex_1 \rightarrow \vertex_2}(\vertex) := \KwNull$;}
                        }
		                \State{$\approximate{\translation}_{\vertex_1 \rightarrow \vertex_2}(\vertex_1) := \vertex_2$;}
		                \State{$score_{\approximate{\translation}_{\vertex_1 \rightarrow \vertex_2}} := 0$;}
                        \For{$i = 1 ~\KwTo \left\lceil\frac{\cardinal{\vertices_1}}{K}\right\rceil$}
		                {
		                    \State{$current\_best_\approximate{\translation} := \KwNull$;}
		                    \State{$current\_best\_score := \infty$;}
		                    \State{$\vertices_1^- := \vertices_1 \setminus \condSet{\vertex \in \vertices_1}{best_{\approximate{\translation}_{\vertex_1 \rightarrow \vertex_2}}(\vertex) \neq \KwNull}$;}
		                    \State{$\vertices_2^- := \vertices_2 \setminus \set{best_{\approximate{\translation}_{\vertex_1 \rightarrow \vertex_2}}(\vertex), \forall \vertex \in \vertices_1} \cup \underbrace{\set{\blackHole, \dots, \blackHole}}_{K~times}$;}
		                    \ForEach{$\sequence{\vertex_{11}, \dots, \vertex_{1K}} \in \underline{permutations}(\vertices_1^-, K)$}
		                    {
		                        \ForEach{$\sequence{\vertex_{21}, \dots, \vertex_{2K}} \in \underline{unique\_permutations}(\vertices_2^-, K)$}
		                        {
    		                            \State{$\approximate{\translation} := \approximate{\translation}_{\vertex_1 \rightarrow \vertex_2}$;}
    		                            \For{$j = 1 ~\KwTo~ K$}
    		                            {
        		                            \State{$\approximate{\translation}[\vertex_{1j}] := \vertex_{2j}$;}
    		                            }
    		                            \State{$score := \underline{\score}(\approximate{\translation})$;}
    		                            \If{$score < current\_best\_score$}
    		                            {
    		                                \State{$current\_best\_score := score$;}
    		                                \State{$current\_best_\approximate{\translation} := \approximate{\translation}$;}
    		                            }
		                        }
		                    }
		                    \State{$\approximate{\translation}_{\vertex_1 \rightarrow \vertex_2} := current\_best_\approximate{\translation}$;}
		                    \State{$score_{\approximate{\translation}_{\vertex_1 \rightarrow \vertex_2}} := current\_best\_score$;}
		                }
                        \Return{$\approximate{\translation}_{\vertex_1 \rightarrow \vertex_2}, score_{\approximate{\translation}_{\vertex_1 \rightarrow \vertex_2}}$;}
	                    \caption{$minimize\_\score~(\vertex_1, \vertex_2, \graph = \tuple{\vertices, \edges}, \vertices_1, \vertices_2)$}
	                    \label{algoGreedy}
	                \end{algorithm}
	                
	                In this algorithm, we decompose assignment of an image to all vertices of $\vertices_1$ into a sequence of assignments of $K$ vertices from that set.
                    Note that the choice of $K$ here controls the complexity of the algorithm.
                    As a matter of fact, when $K = 1$, there are $\cardinal{\vertices_2^-} + 1$ possible assignments for a single vertex.
                    In that case, this leads to a worst-case complexity for $\underline{minimize\_\score}$ of $\complexity{\cardinal{\vertices_1} \cardinal{\vertices_2}}$.
                    On the opposite, when $K = \cardinal{\vertices_1}$, this heuristics becomes an exhaustive search of the best approximate translation $\approximate{\translation}_{\vertex_1 \rightarrow \vertex_2} : \vertices_1 \to \vertices_2 \cup \set{\blackHole}$.

            \section{Approximate translations on random graphs}
            \label{examples}
                
                In this section, we illustrate translation of a localized signal on a random geometric graph using approximate translations\footnote{\label{github}In order to encourage reproducible research, Python codes are provided at \url{https://github.com/BastienPasdeloup/graph_translations}.}.
                In details, we consider the following settings for \algref{algoBestComposition}:
                \begin{itemize}
                    \item The graph is generated randomly using a random geometric model with parameters $\graphOrder = 100$ and $\radius = 0.15$;
                    \item Vertices on which signal is defined (\ie, initial value of $\vertices_1$) are chosen as a random vertex $\vertex_{src}$, and its 1-hop neighborhood;
                    \item Target vertex $\vertex_{tgt}$ is chosen at random in $\vertices \setminus \vertices_1$;
                    \item Selection of $\vertices_2$ on l.20 is made as indicated in the proposed adaptation in \secref{localizedSignals};
                    \item Experiments are performed for all possible combinations of values of $\alpha, \beta, \gamma \in \set{0.1, 0.5, 1}$ (in \eqref{score}), and $K \in \set{1, 2, 3}$ (in \algref{algoGreedy}).
                \end{itemize}
                
                \figref{approxResults} depicts the initial graph on which a signal defined over $5$ vertices --- identified with distinct colors --- is initialized.
                Approximate translations allowing to move from an identified $\vertex_{src}$ to another $\vertex_{tgt}$ that are given are those associated with one Pareto optima in terms of loss and deformation.
                A second Pareto optimum was found with the set of parameters identified above.
                It is given in \figref{approxResults2}.
                Additional examples --- for other graphs and other Pareto optima --- are provided online\footnoteref{github} due to space reasons.

                \begin{figure*}
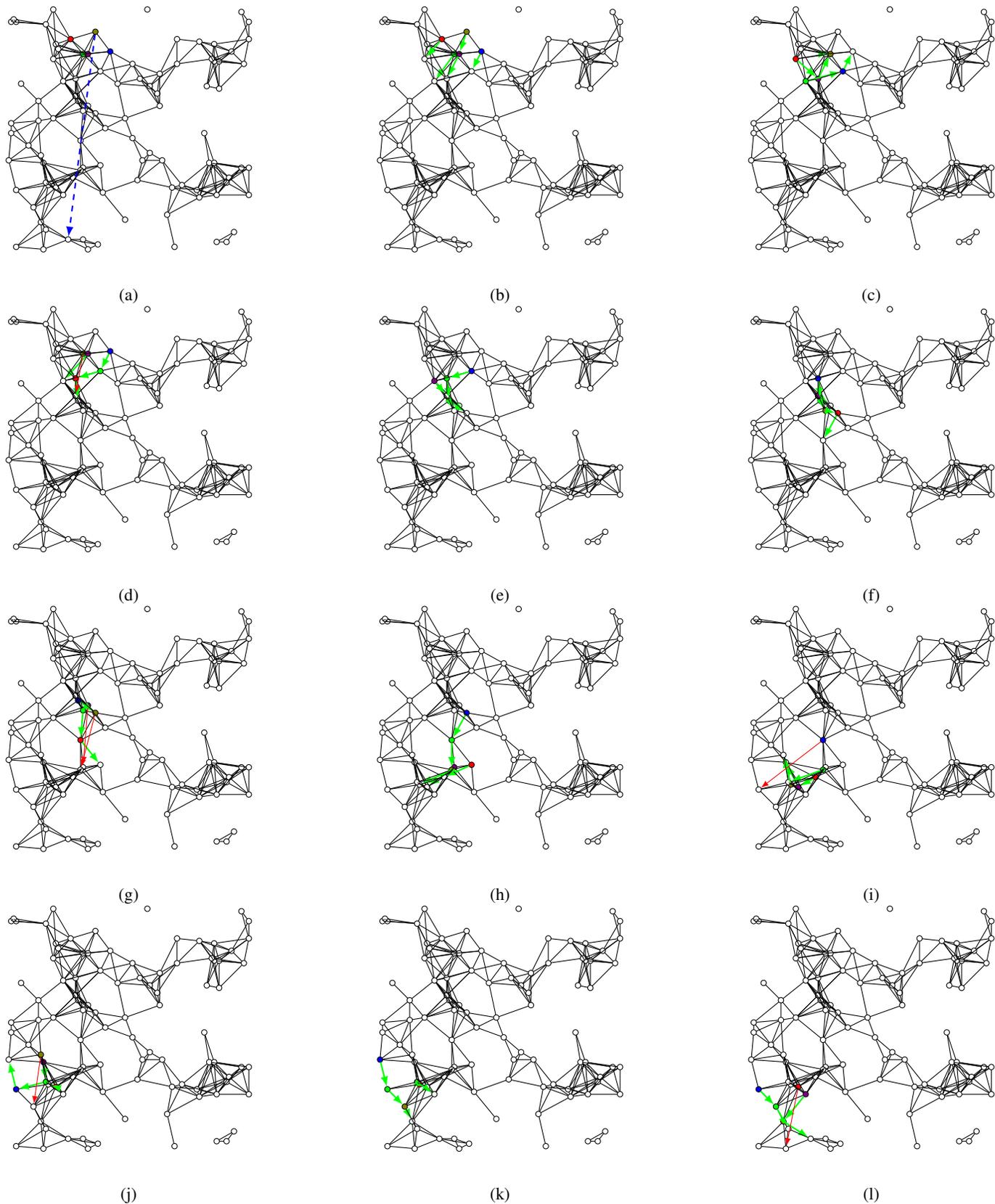

                    \begin{subfigure}{0.25\textwidth}
                        \centering
                        \input{"Objective.tex"}
                        \caption{}
                    \end{subfigure}
                    \hfill
                    \begin{subfigure}{0.25\textwidth}
                        \centering
                        \input{"k_1_alpha_0.5_beta_0.1_gamma_0.5_Step_1.tex"}
                        \caption{}
                    \end{subfigure}
                    \hfill
                    \begin{subfigure}{0.25\textwidth}
                        \centering
                        \input{"k_1_alpha_0.5_beta_0.1_gamma_0.5_Step_2.tex"}
                        \caption{}
                    \end{subfigure}
                    \begin{subfigure}{0.25\textwidth}
                        \centering
                        \input{"k_1_alpha_0.5_beta_0.1_gamma_0.5_Step_3.tex"}
                        \caption{}
                    \end{subfigure}
                    \hfill
                    \begin{subfigure}{0.25\textwidth}
                        \centering
                        \input{"k_1_alpha_0.5_beta_0.1_gamma_0.5_Step_4.tex"}
                        \caption{}
                    \end{subfigure}
                    \hfill
                    \begin{subfigure}{0.25\textwidth}
                        \centering
                        \input{"k_1_alpha_0.5_beta_0.1_gamma_0.5_Step_5.tex"}
                        \caption{}
                    \end{subfigure}
                    \begin{subfigure}{0.25\textwidth}
                        \centering
                        \input{"k_1_alpha_0.5_beta_0.1_gamma_0.5_Step_6.tex"}
                        \caption{}
                    \end{subfigure}
                    \hfill
                    \begin{subfigure}{0.25\textwidth}
                        \centering
                        \input{"k_1_alpha_0.5_beta_0.1_gamma_0.5_Step_7.tex"}
                        \caption{}
                    \end{subfigure}
                    \hfill
                    \begin{subfigure}{0.25\textwidth}
                        \centering
                        \input{"k_1_alpha_0.5_beta_0.1_gamma_0.5_Step_8.tex"}
                        \caption{}
                    \end{subfigure}
                    \begin{subfigure}{0.25\textwidth}
                        \centering
                        \input{"k_1_alpha_0.5_beta_0.1_gamma_0.5_Step_9.tex"}
                        \caption{}
                    \end{subfigure}
                    \hfill
                    \begin{subfigure}{0.25\textwidth}
                        \centering
                        \input{"k_1_alpha_0.5_beta_0.1_gamma_0.5_Step_10.tex"}
                        \caption{}
                    \end{subfigure}
                    \hfill
                    \begin{subfigure}{0.25\textwidth}
                        \centering
                        \input{"k_1_alpha_0.5_beta_0.1_gamma_0.5_Step_11.tex"}
                        \caption{}
                    \end{subfigure}
                    \caption{
                        \textbf{(a)} Initial signal to translate.
                        Vertex $\vertex_{src}$ is located at the source of the blue arrow, and $\vertex_{tgt}$ at its target.
                        \textbf{(b-l)} 11 steps to follow to reach the target.
                        At each step, green arrow indicate that an existing edge is taken (preservation of the EC property), while a red arrow depicts usage of a non-existing edge.
                        The depicted translation corresponds to a Pareto optimum of $10\%$ deformation with zero loss, obtained for parameters $K=1, \alpha = 0.5, \beta = 0.1, \gamma = 0.5$.
                    }
                    \label{approxResults}
                \end{figure*}
                
                \begin{figure*}
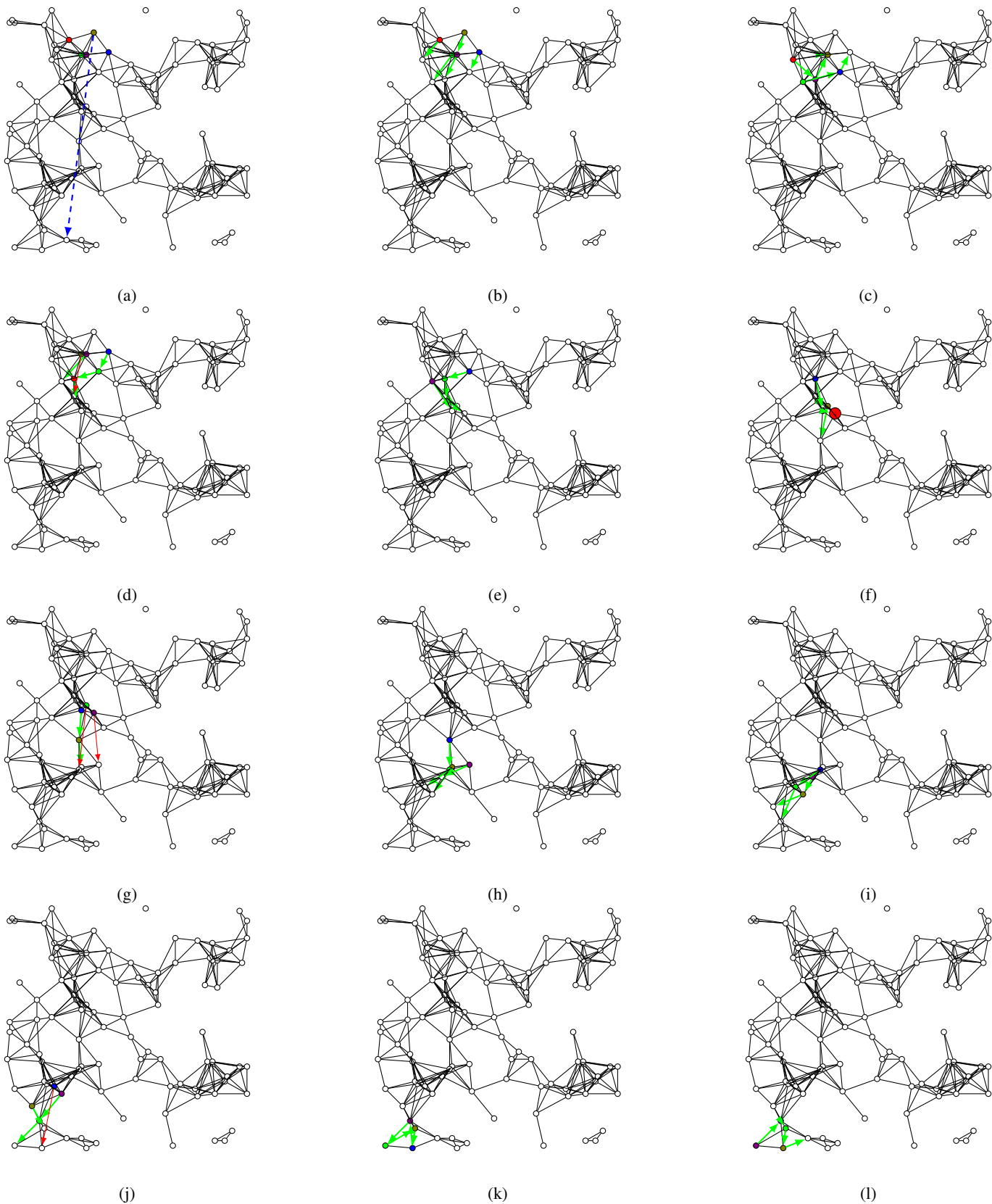

                    \begin{subfigure}{0.25\textwidth}
                        \centering
                        \input{"Objective.tex"}
                        \caption{}
                    \end{subfigure}
                    \hfill
                    \begin{subfigure}{0.25\textwidth}
                        \centering
                        \input{"k_1_alpha_0.5_beta_0.1_gamma_1_Step_1.tex"}
                        \caption{}
                    \end{subfigure}
                    \hfill
                    \begin{subfigure}{0.25\textwidth}
                        \centering
                        \input{"k_1_alpha_0.5_beta_0.1_gamma_1_Step_2.tex"}
                        \caption{}
                    \end{subfigure}
                    \begin{subfigure}{0.25\textwidth}
                        \centering
                        \input{"k_1_alpha_0.5_beta_0.1_gamma_1_Step_3.tex"}
                        \caption{}
                    \end{subfigure}
                    \hfill
                    \begin{subfigure}{0.25\textwidth}
                        \centering
                        \input{"k_1_alpha_0.5_beta_0.1_gamma_1_Step_4.tex"}
                        \caption{}
                    \end{subfigure}
                    \hfill
                    \begin{subfigure}{0.25\textwidth}
                        \centering
                        \input{"k_1_alpha_0.5_beta_0.1_gamma_1_Step_5.tex"}
                        \caption{}
                    \end{subfigure}
                    \begin{subfigure}{0.25\textwidth}
                        \centering
                        \input{"k_1_alpha_0.5_beta_0.1_gamma_1_Step_6.tex"}
                        \caption{}
                    \end{subfigure}
                    \hfill
                    \begin{subfigure}{0.25\textwidth}
                        \centering
                        \input{"k_1_alpha_0.5_beta_0.1_gamma_1_Step_7.tex"}
                        \caption{}
                    \end{subfigure}
                    \hfill
                    \begin{subfigure}{0.25\textwidth}
                        \centering
                        \input{"k_1_alpha_0.5_beta_0.1_gamma_1_Step_8.tex"}
                        \caption{}
                    \end{subfigure}
                    \begin{subfigure}{0.25\textwidth}
                        \centering
                        \input{"k_1_alpha_0.5_beta_0.1_gamma_1_Step_9.tex"}
                        \caption{}
                    \end{subfigure}
                    \hfill
                    \begin{subfigure}{0.25\textwidth}
                        \centering
                        \input{"k_1_alpha_0.5_beta_0.1_gamma_1_Step_10.tex"}
                        \caption{}
                    \end{subfigure}
                    \hfill
                    \begin{subfigure}{0.25\textwidth}
                        \centering
                        \input{"k_1_alpha_0.5_beta_0.1_gamma_1_Step_11.tex"}
                        \caption{}
                    \end{subfigure}
                    \caption{
                        \textbf{(a)} Initial signal to translate.
                        Vertex $\vertex_{src}$ is located at the source of the blue arrow, and $\vertex_{tgt}$ at its target.
                        \textbf{(b-l)} 11 steps to follow to reach the target.
                        At each step, green arrow indicate that an existing edge is taken (preservation of the EC property), while a red arrow depicts usage of a non-existing edge.
                        Additionally, the moment at which loss happens is identified by adding label $\blackHole$ to the lost vertex in \textbf{(f)}.
                        The depicted translation corresponds to a Pareto optimum one lost vertex and no deformation, obtained for parameters $K=1, \alpha = 0.5, \beta = 0.1, \gamma = 1$.
                    }
                    \label{approxResults2}
                \end{figure*}

                Results depicted here show that the method is able to find nice approximate translations when the graph is sufficiently regular, even for small values of $K$.
                When the graph becomes too irregular however --- \eg, using a scale-free graph ---, numerous violations of the translation properties become necessary, and the signal generally suffers from strong losses.
                Very irregular topologies thus remains a challenge, for which we believe new solutions should be developed.

            \section{Conclusions}
            \label{conclusions}
                
                In this article, we have introduced a novel definition for translations on graphs.
                By observing that translations on a toric Euclidean space are exactly lossless translations on a torus graph, we have been able to propose a general notion of translation on arbitrary graphs that preserves neighboring properties.
                Our translations have the property to follow the edges of the graph and, when lossless, guarantee that two neighboring signal entries become located on neighboring vertices after translation of the whole signal.
                Additionally, lossless translations do not change any entry in the signal that is translated, contrary to most existing approaches.
                On a grid graph, translations admitting a loss are exactly those we expect to find when shifting an image on a non-toric Euclidean space.
                We have shown that identification of these translations --- lossless or not --- is an NP-complete problem, and have proposed a relaxation allowing controlled violations of translation properties.

                Our experiments have shown that --- due to their inspiration from translations on Euclidean spaces --- our relaxed operators perform well on graph that are regular enough, such as an irregular manifold sampling, represented by a random geometric graph.
                One of the most important challenges we need to address is therefore considering more complex graphs that feature high irregularities.
                Such graphs are prevalent in real-world networks, such as social networks or brain functional networks, for which few vertices have a large degree and act as hubs.
                In order to provide solutions for such cases, we believe new solutions should be developed, by possibly performing some regularizations on the graphs.
                
            %
            
            
            \bibliographystyle{IEEEtran}
            \bibliography{bibliography}
        

    \end{document}